\documentclass[acmsmall,screen,noacm]{acmart}
\AtBeginDocument{%
  }

\usepackage{subcaption}

\usepackage{wrapfig}

\usepackage{textcomp} 
\usepackage{stmaryrd} 
\SetSymbolFont{stmry}{bold}{U}{stmry}{m}{n}

\usepackage{mathtools}

\usepackage[capitalize,nameinlink]{cleveref} 
\crefformat{definition}{Definition~#2#1#3}
\crefformat{equation}{Equation~(#2#1#3)}
\crefformat{lemma}{Lemma~#2#1#3}
\crefformat{theorem}{Theorem~#2#1#3}
\crefformat{section}{Section~#2#1#3}

\usepackage{thmtools,thm-restate}

\usepackage{url}

\usepackage{mathpartir}

\usepackage{booktabs}
\usepackage{multirow}

\usepackage{xcolor}
\definecolor{ltblue}{rgb}{0,0.4,0.4}
\definecolor{dkblue}{rgb}{0,0.1,0.6}
\definecolor{dkgreen}{rgb}{0,0.35,0}
\definecolor{dkviolet}{rgb}{0.3,0,0.5}
\definecolor{dkred}{rgb}{0.5,0,0}

\usepackage{listings}
\usepackage[cache=false]{minted}

\usepackage{tikz} 
\usetikzlibrary{cd}
\usetikzlibrary{positioning}
\usetikzlibrary{shapes.multipart}
\usetikzlibrary{%
  arrows,%
  shapes.misc,
  shapes.arrows,%
  shapes.callouts,
  chains,%
  matrix,%
  positioning,
  scopes,%
  decorations.pathmorphing,
  decorations.text,
  shadows,%
  quotes%
}
\usetikzlibrary{quantikz2}

\usepackage{xparse}
\NewDocumentCommand{\optionalParens}{s m m}{
    \IfBooleanTF{#2}{\left(#3\right)}{\IfBooleanTF{#1}{~#3}{#3}}
}
\NewDocumentCommand{\apply}{m s O{} m}{
    #1#3 \optionalParens*{#2}{#4}
}

\NewDocumentCommand{\applytwo}{m O{} s m s m}{ 
   #1#2~\optionalParens{#3}{#4}~\optionalParens{#5}{#6}
}

\newcommand{\letin}[3]{\text{let}~#1 = #2; ~#3}

\newcommand{\fix}[3]{\text{fix}~#1.~#2.~#3}

\usepackage{xspace}

\usepackage{cmll} 
\newcommand{\lolli}{\multimap}

\usepackage{braket}
\usepackage[compat=0.6]{yquant}

\usepackage{enumerate}

\newcommand{\MEAS}[1]{\texttt{meas}(#1)}
\newcommand{\IF}[3]{\textrm{if}~#1~\textrm{then}~#2~\textrm{else}~#3}

\newcommand{\cfg}[3]{\langle #1; #2; #3 \rangle}

\newcommand{\Actors}[1]{\text{Actors}(#1)}
\newcommand{\interp}[1]{\llbracket #1 \rrbracket}

\newcommand{\sendin}[3]{\text{send}~#1~\text{to}~#2; ~#3}
\newcommand{\receivefrom}[2]{\letin{#2}{\text{receive from}~#1}}
\newcommand{\establishwith}[2]{\letin{#1}{\text{entangle with}~#2}}
\newcommand{\bit}{\texttt{bit}}
\newcommand{\qubit}{\texttt{qubit}}
\newcommand{\NEW}[1]{\texttt{new}(#1)}
\newcommand{\domain}[1]{\text{dom}(#1)}
\newcommand{\REF}[1]{#1}
\newcommand{\ol}{\overline}
\newcommand{\dom}[1]{\domain{#1}}

\newcommand{\defeq}{\triangleq}

\usetikzlibrary{calc}
\newcommand{\cnotinline}{%
\tikz[baseline=-0.55ex, x=1ex, y=1ex, line width=0.45pt]{%
\draw (0,0.75) -- (0,-0.75); 
\fill (0,0.75) circle (0.18); 
\draw (0,-0.75) circle (0.38); 
\draw (-0.38,-0.75) -- (0.38,-0.75); 
\draw (0,-1.13) -- (0,-0.37); 
}%
}

\acmDOI{XXXXXXX.XXXXXXX}

\setcopyright{none}

\begin{document}

\title{Qoreo: Choreographic Programming for Quantum Distributed Systems}

\author{Jennifer Paykin}
\orcid{0009-0008-9502-3219}
\affiliation{
  \institution{University of Vermont}
  \city{Burlington}
  \state{Vermont}
  \country{USA}
}
\email{jennifer.paykin@uvm.edu}

\author{Steven Baldasty}
\orcid{0009-0001-9617-8203}
\affiliation{
  \institution{University of Vermont}
  \city{Burlington}
  \state{Vermont}
  \country{USA}
}
\email{steven.baldasty@uvm.edu}

\author{Joseph P. Near}
\orcid{0000-0002-3203-3742}
\affiliation{
  \institution{University of Vermont}
  \city{Burlington}
  \state{Vermont}
  \country{USA}
}
\email{jnear@uvm.edu}

\author{Christian Skalka}
\orcid{0000-0002-0402-809X}
\affiliation{
  \institution{University of Vermont}
  \city{Burlington}
  \state{Vermont}
  \country{USA}
}
\email{ceskalka@uvm.edu}


\begin{abstract}
  Programming distributed quantum systems requires multiple actors to
  coordinate precise sequences of quantum operations, classical
  communication, and entanglement generation.
  Writing such protocols
  directly as distributed processes is tedious and error-prone, and
  subtle mismatches can cause deadlock or silently incorrect quantum
  states.
  We present \emph{Qoreo}, a choreographic programming language for
  quantum distributed systems in which an entire protocol is expressed
  as single, global program (a choreography) rather than as a collection of independent
  actor processes.
  Qoreo includes a local quantum language with linear types that
  enforce the no-cloning principle; a choreographic language that
  combines local quantum computation with inter-actor classical and
  quantum communication; and a process language for individual network
  nodes.
  We prove type safety for choreographies, guaranteeing that
  well-typed programs implement well-defined quantum operations, and we
  define endpoint projection~(EPP), which automatically derives a
  network of independent processes from any choreography.
  We prove EPP sound and complete with respect to the choreographic
  semantics; as a corollary, every well-typed choreography projects to
  a deadlock-free process network.
  The metatheory of Qoreo is fully mechanized in Rocq, and we provide
  an extraction pipeline to NetQASM for simulation and deployment on
  quantum network hardware.
\end{abstract}


\maketitle

\section{Introduction}
\label{sec:introduction}

Quantum communication protocols and distributed systems are an increasingly important part of quantum computing as the demand for useful quantum computing systems grows. Distributed protocols support scalable computing and allow for secure communication of both quantum and non-quantum (classical) information~\citep{cao_evolution_2022,caleffi_distributed_2024,zheng_large-scale_2026}.

As an example, consider the quantum teleportation protocol, shown in \cref{fig:teleportation-circuit}. The goal of quantum teleportation is for Alice to communicate the state of her qubit $q$ to Bob. Alice cannot physically send Bob the qubit; among other concerns, it would be prohibitively error-prone over long distances.
Instead, Alice and Bob rely on the \emph{quantum entanglement} to communicate the state. We assume that Alice and Bob can establish a pair of entangled qubits, indicated by the \text{entanglement} box.

This style of presentation is a generalization of the quantum circuit model, where primitive quantum operations (such as measurement, the $X$/$Z$/$H$ gates, as well as the two-qubit controlled-not gate \cnotinline) can be applied to qubits, represented as wires. For distributed protocols like teleportation, we indicate a visual separation between different actors in the distributed system and informally check that the only types of interaction between them correspond to either entanglement generation or classical communication.

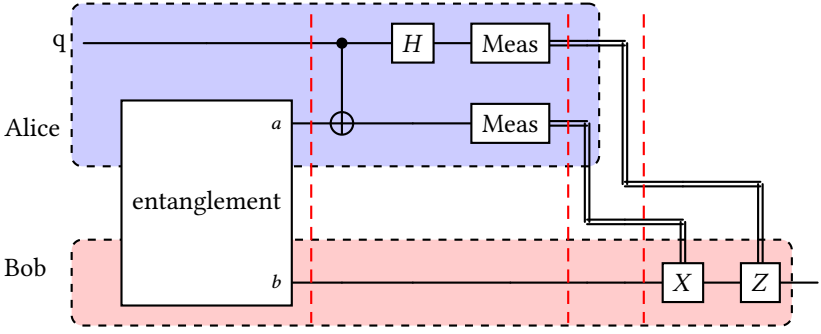
\begin{figure}
    \begin{quantikz}
        \gategroup[2,steps=6,style={dashed,rounded corners,fill=blue!20, inner xsep=1pt}, background,label style={label position=left,anchor=west,xshift=-1cm,yshift=-0.75cm}]{Alice}
        \lstick{q~~~~~~}    & \slice{}     & \ctrl{1}                 & \gate{H} & \gate{\text{Meas}} \slice{} & \setwiretype{c}                 
            & \wire[d][2]{c}   \\
        \setwiretype{n}  & \gate[4]{\text{entanglement}} \gateoutput{$a$}
            & \targ{} \setwiretype{q}  &          & \gate{\text{Meas}}                & \wire[d][2]{c} \setwiretype{c}  
            & \wireoverride{n} \\
        \setwiretype{n}  &                           
            &                          &          &                                   &
            &                          & \setwiretype{c}        &                 \\
        \setwiretype{n}  &
            &                          &          &                                   & 
            & \setwiretype{c}  &               \\
        \gategroup[1,steps=9,style={dashed,rounded corners,fill=red!20, inner xsep=1pt}, background,label style={label position=left,anchor=west,xshift=-1cm,yshift=0cm}]{Bob}
        \setwiretype{n}  &  \gateoutput{$b$}
            & \setwiretype{q}          &          &                                   &                                 
            &   \slice{}  &          \gate{X} \wire[u][1]{c} & \gate{Z} \wire[u][2]{c} & 
    \end{quantikz}
    \caption{Quantum teleportation protocol as an informal quantum circuit diagram. Alice starts with an arbitrary qubit $q$, and by the end of the protocol, Bob should have a qubit in that same state. Dashed vertical red lines separate the protocol into four stages. Stage 1: Alice and Bob establish a pair of entangled qubits $a$ and $b$. Stage 2: Alice performs some local operations (controlled-not, hadamard (H)) to entangle her starting qubit $q$ with $a$; she then measures both qubits, resulting in two classical bits (indicated with double wires). Stage 3: Alice sends the two classical bits to Bob over ordinary classical channels. Stage 4: Bob uses the classical messages to correct his remaining qubit by conditionally applying unitary $X$ and $Z$ gates.}
    \label{fig:teleportation-circuit}
\end{figure}



To implement protocols like these on distributed quantum systems, programmers translate them into \emph{distributed processes}, where each actor in the protocol executes an independent process that interacts with other actors via entanglement generation and classical message-passing. Frameworks like NetSQUID, NetQASM, and others allow users to express processes and either simulate or deploy them to hardware~\citep{coopmans_netsquid_2021,dahlberg_netqasm_2022,cardama_netqmpi_2025}. Consider \Cref{fig:teleportation-processes}, which shows processes (defined later in \Cref{sec:network}) implementing Alice's and Bob's roles in the quantum teleportation protocol.

Unfortunately, programming with distributed processes can be repetitive and error-prone, not to mention difficult to debug and analyze. Suppose Alice forgets to send $z$ to Bob after sending $x$. Then Bob's process will be deadlocked---he will be unable to proceed as he will be waiting for Alice's message that will never come.
An even more subtle error could occur if Bob thinks that Alice is going to send $z$ before $x$. In this case, Bob won't know anything went wrong even at runtime; he would just continue his procedure with incorrect values. 

Researchers have explored a few different techniques to try to limit the possibility of such errors. 
Quantum process calculi including communicating quantum processes~\cite{gay_2005_communicating} have explored formal verification of correctness properties using model checking~\citep{davidson_2012_model}. \citet{lanese_towards_2025} integrate multiparty session types with quantum processes~\citep{lanese_towards_2025} to ensure that processes follow a protocol specified by a global type. \citet{buckley_algebraic_2024} use Kleene algebras to analyze process specifications to show how quantum data is distributed through a network.

\begin{figure}
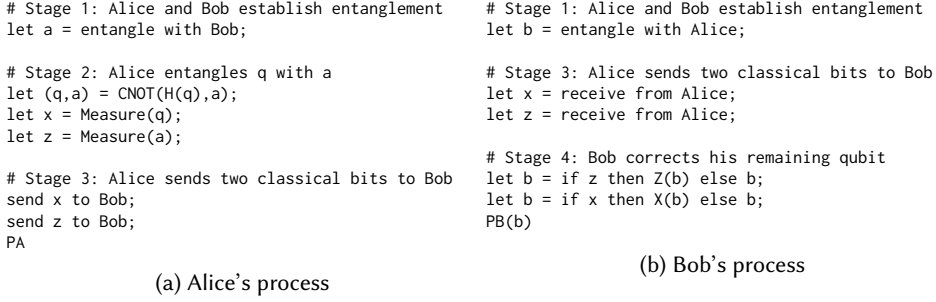

    \centering
\begin{subfigure}[t]{0.45\textwidth} 
\scriptsize
\begin{verbatim}
# Stage 1: Alice and Bob establish entanglement
let a = entangle with Bob;

# Stage 2: Alice entangles q with a
let (q,a) = CNOT(H(q),a);
let x = Measure(q);
let z = Measure(a);

# Stage 3: Alice sends two classical bits to Bob
send x to Bob;
send z to Bob;
PA
\end{verbatim}
\caption{Alice's process}
\end{subfigure}
\begin{subfigure}[t]{0.45\textwidth} 
\scriptsize
\begin{verbatim}
# Stage 1: Alice and Bob establish entanglement
let b = entangle with Alice;

# Stage 3: Alice sends two classical bits to Bob
let x = receive from Alice;
let z = receive from Alice;

# Stage 4: Bob corrects his remaining qubit
let b = if z then Z(b) else b;
let b = if x then X(b) else b;
PB(b)
\end{verbatim}
\caption{Bob's process}
\end{subfigure}
    \caption{Implementation of the quantum teleportation protocol in the quantum process model. The syntax of processes is introduced in \Cref{sec:network}: Alice's operation \texttt{entangle with Bob} establishes an entangled pair when combined with Bob's corresponding \texttt{entangle with Alice}. \texttt{Meas} refers to quantum measurement, and \texttt{CNOT}, \texttt{H}, \texttt{X}, and \texttt{Z} are standard quantum gates.
    \texttt{PA} and \texttt{PB} are continuations---processes that Alice and Bob will execute after the teleportation protocol.
    }
    \label{fig:teleportation-processes}
\end{figure}

In this paper we propose a different approach inspired by \emph{choreographic programming}~\cite{carbone2013deadlock, montesi2023introduction}. A choreography is a global view of the behaviors of all actors in a distributed protocol that allows us to write programs in the style of the informal, top-level circuit diagrams such as the one in \Cref{fig:teleportation-circuit}. The choreography for quantum teleportation is shown in \Cref{fig:teleportation}, and bears a striking similarity to the informal description. We argue that the choreographic programming model is practical and programmatic from a user perspective compared to the process model, and in many cases are easier to reason about due to its global view of the system.

Furthermore, choreographies can be automatically projected to processes running in the process model. This procedure, called endpoint projection (EPP), comes with the guarantee that the semantics of the projected network is the same as the semantics of the choreography. As a consequence, any reasoning we can do about the choreographic language (taking advantage of the global view of the protocol) is also true of the projected processes. Concretely, the projected network is deadlock-free because the choreography is.

\begin{wrapfigure}{R}{0.5\textwidth}
\begin{center}
\scriptsize
\begin{verbatim}
# Stage 1: Alice and Bob establish entanglement
A.a <~> B.b;

# Stage 2: Alice entangles q with a
let A.(q,a) = CNOT(H(q),a);
let A.x = Measure(q);
let A.z = Measure(a);

# Stage 3: Alice sends classical bits to Bob
A.x ~> B.x;
A.z ~> B.z;

# Stage 4: Bob corrects his remaining qubit
let B.b = if z then Z(b) else b;
let B.b = if x then X(b) else b;
C(B.b)
\end{verbatim}
\end{center}
\caption{Quantum teleportation as a choreography. On the last line of the protocol, $C$ is a continuation choreography that uses the qubit $B.b$.}
\label{fig:teleportation}
\end{wrapfigure}

This paper introduces \emph{Qoreo}, a choreographic programming language for distributed quantum systems. Qoreo is made up of a local language for quantum computations along with separate languages for choreographies and distributed processes. The local language is based on the quantum lambda calculus~\cite{gay_quantum_2009}, a linearly-typed higher-order functional language with quantum gates. Choreographies support both local computations as well as communication between actors. 
We introduce a novel linear type system for choreographies and prove type safety, which implies well-typed choreographies implement well-defined quantum operations. We further define endpoint projection and prove soundness and completeness of the projected networks with well-typed source choreographies. As a consequence, we conclude that well-typed choreographies project onto deadlock-free and well-defined quantum processes. Finally, we develop a toolchain for Qoreo protocol simulation in NetQASM to demonstrate practical applicability of the system.


\paragraph{Overview of the paper.}
In Section \ref{sec:background} we discuss the state-of-the-art in
distributed quantum computing including existing implementations that
inform our work.
In  \Cref{sec:quantum} we define our local quantum language and
type theory. We formulate a type safety result for the local language
in Theorem \ref{thm:e.safety}. In Section \ref{sec:choreography} we
develop the full Qoreo choreographic language and type theory
supporting both classical and quantum communication.  We formulate a type
safety result for the choreographic language in
Theorem \ref{thm:C.safety}. In Section \ref{sec:network} we define the
network language, and in Section \ref{sec:epp} we define endpoint
projection and show that it is sound in Theorem \ref{thm:epp-sound}
and complete in Theorem \ref{thm:epp-complete}, with network safety
immediately following as Corollary \ref{cor:net.safety}.  All
definitions and results in Sections \ref{sec:quantum},
\ref{sec:choreography}, and \ref{sec:network} are fully formalized in Rocq
and are available on GitHub\footnote{\url{github.com/uvm-plaid/qoreo}}
In Section \ref{sec:extraction} we discuss the implementation of Qoreo
via extraction from Rocq to NetQASM-- the full pipeline from Qoreo to
NetQASM simulation is also provided in in supplementary materials.  In
Section \ref{sec:examples} we demonstrate practical applications of
Qoreo to examples including distributed control gates and the B92 quantum key
distribution protocol. In Section \ref{sec:related}
we make a detailed comparison with related work, though discussions
of related work also occur in context throughout the paper.
In Section \ref{sec:conclusion} we conclude with comments on future work.

\section{Background and Related Work}
\label{sec:background}
\subsection{Quantum semantics}
\label{sec:background:quantum}

The state of a qubit in a quantum computer is modeled as a superposition of the classical bits $0$ and $1$, written $\alpha \ket{0} + \beta \ket{1}$ where $\alpha,\beta \in \mathbb{C}$ and $|\alpha|^2 + |\beta|^2=1$. 
Similarly, an $n$-qubit quantum state is modeled as a superposition of the $2^n$ bitvectors of length $n$, where the sum of the squared norms of all the coefficients also sums to $1$.
This means that the size of the state space grows exponentially in the number of qubits, which is one of the core semantic advantages of quantum programming.

Quantum computation proceeds by applying unitary transformations, also called unitary gates, to the global state, interleaved with measurements that produce classical outcomes and update the post-measurement state. Measurement is probabilistic---measuring a qubit in state $\alpha \ket{0} + \beta \ket{1}$ will produce $0$ with probability $|\alpha|^2$ and $1$ with probability $|\beta|^2$.


Note that even in the distributed setting, qubits can be entangled, meaning that the global state cannot in general be decomposed into independent local states.
An example is the 2-qubit Bell state $\ket{\Phi^+} = \frac{1}{\sqrt{2}}(\ket{00} + \ket{11})$. The two qubits are entangled in that, if Alice measures the left qubit and observes $0$, then Bob's qubit is necessarily also $\ket{0}$; and similarly with $1$.

\paragraph{Density matrix semantics}

For reasoning about measurement and mixed behavior, we use partial density matrices~\cite{nielsen_quantum_2010} rather than pure states as described above. A partial density matrix on $n$ qubits is a $2^n \times 2^n$ complex matrix $\rho$ satisfying certain properties ($\rho = \rho^\dagger$, $\rho \succeq 0$, and $\mathrm{tr}(\rho) \le 1$). If we think of a pure state $\alpha \ket{0} + \beta \ket{1}$ as a state vector $\ket{\varphi} = \begin{psmallmatrix} \alpha \\ \beta \end{psmallmatrix}$, then its corresponding density matrix representation is $\ket{\varphi} \bra{\varphi} = \begin{psmallmatrix} \alpha \\ \beta \end{psmallmatrix} \begin{psmallmatrix} \alpha^\ast & \beta^\ast \end{psmallmatrix}$.

The trace of a partial density matrix corresponds to the probability of being in that state. Consider the example of measuring the two-qubit Bell state. With probability $1/2$, we will obtain the pure state $\ket{00}$, corresponding to the partial density matrix $\frac{1}{2} \ket{00}\bra{00} = \frac12 \begin{psmallmatrix} 1 & 0 \\ 0 & 0 \end{psmallmatrix}$; with the remaining $1/2$ probability, we will obtain the pure state $\ket{11}$, corresponding to the partial density matrix $\frac{1}{2} \ket{11}\bra{11} = \frac12 \begin{psmallmatrix} 0 & 0 \\ 0 & 1 \end{psmallmatrix}$

This density-matrix presentation is convenient for programming language semantics~\cite{selinger_towards_2004} as it eliminates the need for a probabilistic semantics, replacing it instead with a nondeterministic semantics over partial density matrices.

\subsection{Entanglement generation and simulating quantum distributed systems}

What does it mean to communicate quantum information between actors in a distributed system? The quantum teleportation algorithm shows us that it is possible to ``send'' quantum information, but doing so requires both classical communication and establishing entanglement between distant actors. Other algorithms can be implemented more efficiently without teleportation---for example, applying a distributed two-qubit unitary (see \cref{sec:example:controlled}) requires only one bit of classical communication in each direction rather than teleportation's two~\citep{caleffi_distributed_2024}.

Instead, the accepted quantum communication primitive is entanglement generation---generating qubits in the entangled Bell state $\ket{\Phi^+}=\frac12(\ket{00}+\ket{11})$~\citep{kozlowski_designing_2020, coopmans_netsquid_2021, azuma_quantum_2023}. Remote entanglement generation is imagined as one of the foundational services provided by a quantum network~\cite{wehner2018quantum}.
Note that unlike classical communication, entanglement generation is undirected---it is not a message from Alice to Bob, or vice versa.

\subsection{Linear types and quantum programming}

In the quantum teleportation algorithm, note that after Alice makes her measurements and sends classical bits to Bob, the qubits $A.a$ and $B.b$ are no longer entangled. In fact, $A.a$ has been measured, so it is in a collapsed (classical) state. Quantum communication protocols often assume that a qubit is no longer accessible after it has been measured, and thus the network is free to reallocate that qubit for future entanglement generation. There are other restrictions as well---two-qubit unitaries cannot be applied to the same qubit twice, as the \emph{no-cloning} principle tells us that the state of a qubit cannot be duplicated. Similarly, the \emph{no-discarding} principle tells us that resetting a qubit to a classical state requires measurement, and thus may have a global effect on the quantum state.
Quantum programming languages like the quantum lambda calculus (on which we base our local quantum language in \Cref{sec:quantum}), use linear or substructural type systems to enforce the no-cloning and no-discarding properties~\citep{gay_quantum_2009,green_quipper_2013,koch_imperative_2025}, including in the setting of quantum networks via multiparty session types~\citep{lanese_towards_2025}. The relationship between choreographies and linear logic has been explored~\citep{carbone_choreographies_2018} but as far as we know this work would be the first to integrate a linear type system with choreographies.




\section{The local quantum language}
\label{sec:quantum}
To define Qoreo we begin by presenting a simple quantum programming language based on the quantum lambda calculus~\cite{gay_quantum_2009} that we assume can be executed on individual nodes in the quantum distributed system. In practice this language could be compiled and optimized further, or indeed replaced by another quantum language.

The advantage of the quantum lambda calculus is its strong linear type system and well-defined semantics.
In the quantum lambda calculus, quantum data like qubits and tuples of qubits have linear types while classical data like bits have non-linear types and thus can be duplicated. Linear types are used to ensure that well-typed programs correspond to valid quantum computations. In particular, they ensure that well-typed programs do not try to violate the no-cloning principle of quantum mechanics. The no-cloning principle is a consequence of the fact that quantum computing is built on linear transformations---it says that the state of an arbitrary qubit $\ket{\varphi}$ cannot be cloned to produce a pair of separable qubits $\ket{\varphi} \otimes \ket{\varphi}$. In practice, this means that we cannot apply a two-qubit gate or quantum function to the same qubit twice: $g(q,q)$ is not allowed.

\begin{figure}
    \centering
\begin{align*}
    e &::= x \mid \letin{x}{e}{e'} \mid !e \mid \letin{!x}{e}{e'} \\
    &\mid b \mid \IF{e}{e_1}{e_2} \mid (e_1,e_2) \mid \letin{(x_1,x_2)}{e}{e'} \\
    &\mid \lambda x.e \mid \fix{f}{x}{e} \mid e~e' 
    \mid \NEW{e} \mid \MEAS{e} \mid \REF{q} \mid U~e \tag{expressions} 
    \\
    b &\in \{0,1\} \tag{bits} \\
    x &\in \mathcal{V} \tag{variables} \\
    q &\in \mathcal{Q} \tag{qubit references} \\
    U &::= X \mid Y \mid Z \mid H \mid S \mid S^\dag \mid T \mid T^\dag \mid \textrm{CNOT} \mid \textrm{SWAP} \tag{unitary gates}
\end{align*}
    \caption{Qoreo syntax}
    \label{fig:syntax}
\end{figure}

The syntax is defined in \Cref{fig:syntax}.
We assume we have a built-in set of universal gates (in this case Clifford+T).
The main difference between this language and the original presentation of the quantum lambda calculus is that we syntactically distinguish qubit references from free variables. Qubit references point to an index into the quantum state; they are considered values.

\newcommand{\doublerule}[1][.4pt]{%
  \noindent
  \makebox[0pt][l]{\rule[.7ex]{\linewidth}{#1}}%
  \rule[.3ex]{\linewidth}{#1}}

\begin{figure}
\footnotesize
\begin{mathpar}
~
    \inferrule*[right=letC]
    {\cfg{e_1}{\Theta}{\rho} \rightarrow \cfg{e_1'}{\Theta'}{\rho'}}
    {\cfg{\letin{x}{e_1}{e_2}}{\Theta}{\rho} 
        \rightarrow 
     \cfg{\letin{x}{e_1'}{e_2}}{\Theta'}{\rho'}}

    \inferrule*[right=let!C]
    {\cfg{e_1}{\Theta}{\rho} \rightarrow \cfg{e_1'}{\Theta'}{\rho'}}
    {\cfg{\letin{!x}{e_1}{e_2}}{\Theta}{\rho} 
        \rightarrow 
     \cfg{\letin{!x}{e_1'}{e_2}}{\Theta'}{\rho'}}

    \inferrule*[right=ifC]
    {
        \cfg{e}{\Theta}{\rho} \rightarrow \cfg{e'}{\Theta'}{\rho'}
    }
    {
        \cfg{\IF{e}{e_1}{e_0}}{\Theta}{\rho} 
        \rightarrow 
        \cfg{\IF{e'}{e_1}{e_0}}{\Theta'}{\rho'}
    }

    \newline

    \inferrule*[right=pairC1]
    {
        \cfg{e_1}{\Theta}{\rho}
        \rightarrow
        \cfg{e_1'}{\Theta'}{\rho'}
    }
    {
        \cfg{(e_1,e_2)}{\Theta}{\rho} 
        \rightarrow
        \cfg{(e_1',e_2)}{\Theta'}{\rho'}
    }

    \inferrule*[right=pairC2]
    {
        \cfg{e_2}{\Theta}{\rho}
        \rightarrow
        \cfg{e_2'}{\Theta'}{\rho'}
    }
    {
        \cfg{(v_1,e_2)}{\Theta}{\rho} 
        \rightarrow
        \cfg{(v_1,e_2')}{\Theta'}{\rho'}
    }

    \inferrule*[right=letPairC]
    {
        \cfg{e}{\Theta}{\rho}
        \rightarrow
        \cfg{e'}{\Theta'}{\rho'}
    }
    {
        \cfg{\letin{(x_1,x_2)}{e}{e''}}{\Theta}{\rho} 
        \rightarrow 
        \cfg{\letin{(x_1,x_2)}{e'}{e''}}{\Theta'}{\rho'}
    }

    \inferrule*[right=appC1]
    {
        \cfg{e}{\Theta}{\rho}
        \rightarrow
        \cfg{e'}{\Theta'}{\rho'}
    }
    {
        \cfg{e~e''}{\Theta}{\rho}
        \rightarrow
        \cfg{e'~e''}{\Theta'}{\rho'}
    }

    \inferrule*[right=appC2]
    {
        \cfg{e}{\Theta}{\rho}
        \rightarrow
        \cfg{e'}{\Theta'}{\rho'}
    }
    {
        \cfg{v~e}{\Theta}{\rho}
        \rightarrow
        \cfg{v~e'}{\Theta'}{\rho'}
    }

    \inferrule*[right=newC]
    { \cfg{e}{\Theta}{\rho}
      \rightarrow
      \cfg{e'}{\Theta'}{\rho'}
    }
    {
        \cfg{\NEW{e}}{\Theta}{\rho}
        \rightarrow
        \cfg{\NEW{e'}}{\Theta'}{\rho'}
    }

    \inferrule*[right=measC]
    { \cfg{e}{\Theta}{\rho}
      \rightarrow
      \cfg{e'}{\Theta'}{\rho'}
    }
    {
        \cfg{\MEAS{e}}{\Theta}{\rho}
        \rightarrow
        \cfg{\MEAS{e'}}{\Theta'}{\rho'}
    }

    \inferrule*[right=UC]
    { \cfg{e}{\Theta}{\rho}
      \rightarrow
      \cfg{e'}{\Theta'}{\rho'}
    }
    {
        \cfg{U~e}{\Theta}{\rho}
        \rightarrow
        \cfg{U~e'}{\Theta'}{\rho'}
    }
    
\doublerule
\end{mathpar}
\begin{align*}
    \cfg{\letin{x}{v}{e'}}{\Theta}{\rho} 
        &\rightarrow 
    \cfg{e'[v/x]}{\Theta}{\rho}
    &\textsc{letB}
    \\
    \cfg{\letin{!x}{!e}{e'}}{\Theta}{\rho} 
         &\rightarrow 
      \cfg{e'[e/x]}{\Theta}{\rho}
      &\textsc{let!B}
    \\
    \cfg{\IF{b}{e_1}{e_0}}{\Theta}{\rho}
         &\rightarrow
         \cfg{e_b}{\Theta}{\rho}
         &\textsc{ifB} 
    \\
    \cfg{\letin{(x_1,x_2)}{(v_1,v_2)}{e'}}{\Theta}{\rho} 
         &\rightarrow 
         \cfg{e'[e_1/x_1, e_2/x_2]}{\Theta}{\rho}
          &\textsc{letPairB}
    \\
    \cfg{(\lambda x.e')~v}{\Theta}{\rho}
            &\rightarrow
            \cfg{e'[v/x]}{\Theta}{\rho}
            &\textsc{appB}
    \\
    \cfg{(\fix{f}{x}{e'})~(!e)}{\Theta}{\rho}
            &\rightarrow
            \cfg{e'[e/x,(\fix{f}{x}{e'})/f]}{\Theta}{\rho}
            &\textsc{fixB}
    %
    %
    %
    %
    %
\end{align*}
\begin{mathpar}
\doublerule

    \inferrule*[right=newB]
    {
        q \notin \domain{\Theta} \\
        n = \text{num qubits of}~\rho
    }
    {\cfg{\NEW{b}}{\Theta}{\rho} \rightarrow \cfg{\REF{q}}{ \left(\Theta,q \mapsto n\right)}{\rho \otimes \ket{b}\bra{b}}
    }
    
    \inferrule*[right=UB]
    {
        \Theta(q_i) = j_i \\
        \rho' = \interp{U}_{j_1,\ldots,j_k} \rho \interp{U}_{j_1,\ldots,j_k}^\dagger
    }
    {
        \cfg{U(\REF{q_1},\ldots,\REF{q_k})}{\Theta}{\rho}
        \rightarrow
        \cfg{(\REF{q_1},\ldots,\REF{q_k})}{\Theta}{\rho'}
    }

    \inferrule*[right=measB0]
    {
      \Theta(q) = i \\
      \Theta' = \Theta \setminus \{q\} \\
      \rho' = P^i_0 \rho P^i_0 
    }
    { \cfg{\MEAS{\REF{q}}}{\Theta}{\rho}
      \rightarrow
      \cfg{0}{\Theta'}{\rho'}
    }

    \inferrule*[right=measB1]
    {
      \Theta(q) = i \\
      \Theta' = \Theta \setminus \{q\} \\
      \rho' = P^i_1 \rho P^i_1 
    }
    { \cfg{\MEAS{\REF{q}}}{\Theta}{\rho}
      \rightarrow
      \cfg{1}{\Theta'}{\rho'}
    }
\end{mathpar}
\caption{Step relation rules for the local quantum language. The first group consists of contextual reduction rules (\textsc{C}). The second group consists of the usual $\beta$-reduction rules (\textsc{B}) for a call-by-value linear $\lambda$-calculus. Below the line are quantum-specific $\beta$-reduction rules. The syntax $e'[e/x]$ refers to capture-avoiding substitution defined in the usual way. 
}
\label{fig:local-step}
\end{figure}

\begin{figure*}
    \centering
    \footnotesize
    \begin{mathpar}
~
    \inferrule*[right=lvar]
    {x \notin \Gamma}
    {\Gamma;x:\tau;\emptyset \vdash x : \tau}
    
    \inferrule*[right=cvar]
    {(x:\tau) \in \Gamma}
    {\Gamma;\emptyset;\emptyset \vdash x : \tau}

    \inferrule*[right=$!$]
    {
        \Gamma;\emptyset;\emptyset \vdash e : \tau
    }
    { \Gamma; \emptyset; \emptyset \vdash !e : !\tau }

    \inferrule*[right=bit]
    {b \in \{0,1\} }
    {\Gamma;\emptyset;\emptyset \vdash b : \bit}

    \inferrule*[right=let]
    {
        \Gamma;\Delta_1;\Theta_1 \vdash e : \tau 
        \\
        \Gamma;\Delta_2,x:\tau;\Theta_2 \vdash e' : \tau' 
        \\
        \Delta_1 \bot \Delta_2 \\
        \Theta_1 \bot \Theta_2 \\
        x \notin \domain{\Gamma,\Delta_2}
    }
    {\Gamma;\Delta_1, \Delta_2; \Theta_1,\Theta_2 \vdash \letin{x}{e}{e'} : \tau'}

    \inferrule*[right=let-$!$]
    { \Gamma; \Delta_1;\Theta_1 \vdash e : !\tau \\
      \Gamma,x:\tau ; \Delta_2; \Theta_2 \vdash e' : \tau'  \\
      \Delta_1 \bot \Delta_2 \\
      \Theta_1 \bot \Theta_2 \\
      x~\notin \domain{\Delta_2}
    }
    {\Gamma; \Delta_1,\Delta_2; \Theta_1,\Theta_2 \vdash \letin{!x}{e}{e'} : \tau'}

    \inferrule*[right=if]
    {
        \Gamma;\Delta;\Theta \vdash e : \bit \\
        \Gamma;\Delta';\Theta' \vdash e_1 : \tau \\
        \Gamma;\Delta';\Theta' \vdash e_0 : \tau \\
        \Delta \bot \Delta' \\
        \Theta \bot \Theta'
    }
    {\Gamma;\Delta,\Delta';\Theta,\Theta' \vdash \IF{e}{e_1}{e_0} : \tau}

    \inferrule*[right=$\otimes$]
    {
        \Gamma; \Delta_1; \Theta_1 \vdash e_1 : \tau_1 \\
        \Gamma; \Delta_2; \Theta_2 \vdash e_2 : \tau_2 \\
        \Delta_1 \bot \Delta_2 \\
        \Theta_1 \bot \Theta_2
    }
    { \Gamma; \Delta_1,\Delta_2; \Theta_1,\Theta_2 \vdash (e_1,e_2) : \tau_1 \otimes \tau_2}

    \inferrule*[right=let-$\otimes$]
    {
        \Gamma;\Delta;\Theta \vdash e : \tau_1 \otimes \tau_2 \\
        \Gamma;\Delta',x_1:\tau_1,x_2:\tau_2;\Theta' \vdash e' : \tau' \\
        \Delta \bot \Delta' \\
        \Theta \bot \Theta' \\
         x_1,x_2~\notin \dom{\Delta'},\dom{\Theta'}
    }
    { \Gamma; \Delta,\Delta'; \Theta, \Theta' \vdash \letin{(x_1,x_2)}{e}{e'} : \tau' }

        \inferrule*[right=lambda]
    {
        \Gamma;\Delta,x:\tau_1;\Theta \vdash e : \tau_2 \\
        x \notin \domain{\Gamma,\Delta}
    }
    { \Gamma;\Delta;\Theta \vdash \lambda x.e : \tau_1 \lolli \tau_2 }
    
    \inferrule*[right=fix]
    {
        \Gamma,f:!\tau_1 \lolli \tau_2, x:\tau_2;x:\tau_1; \emptyset \vdash e : \tau_2 \\
        f,x \not\in \domain{\Delta}
    }
    { \Gamma;\emptyset;\emptyset \vdash \fix{f}{x}{e} : !\tau_1 \lolli \tau_2 }

    \inferrule*[right=app]
    {
        \Gamma;\Delta;\Theta \vdash e : \tau_1 \lolli \tau_2 \\
        \Gamma;\Delta';\Theta' \vdash e' : \tau_1 \\
        \Delta \bot \Delta' \\
        \Theta \bot \Theta'
    }
    {\Gamma;\Delta,\Delta';\Theta,\Theta' \vdash e e' : \tau_2}

\doublerule

    \inferrule*[right=meas]
    { \Gamma;\Delta;\Theta : e : \qubit }
    { \Gamma; \Delta;\Theta \vdash \MEAS{e} : \bit }

    \inferrule*[right=qref]
    {~}
    { \Gamma;\emptyset;q \vdash \REF{q} : \qubit}

    \inferrule*[right=new]
    {\Gamma;\Delta;\Theta \vdash e : \bit}
    {\Gamma;\Delta;\Theta \vdash \NEW{e} : \qubit }

    \inferrule*[right=unitary]
    {\Gamma;\Delta;\Theta \vdash e : \tau \\
     \text{type-of}(U) = \tau }
    {\Gamma;\Delta;\Theta \vdash U~e : \tau}
~
    \end{mathpar}
    \caption{Typing rules for the local quantum language. Above the line are rules for the linear lambda calculus, and below the line are quantum-specific typing rules. We write $\Delta_1 \bot \Delta_2$ to say the domains of $\Delta_1$ and $\Delta_2$ are disjoint, and similarly for $\Theta_1 \bot \Theta_2$. 
    }
    \label{fig:local-types}
\end{figure*}

Throughout this paper, we will use finite maps to manage both free variables and qubit references, for which we use the following conventions.

\begin{definition}
\label{def:mapping}
Let $M_1,M_2$ be finite maps with keys $k$. We write $M_1,M_2$ to mean the
concatenation of the two maps, e.g. $M$ such that $M(k) = M_2(k)$ if $k \in \dom{M_2}$, otherwise $M(k) = M_1(k)$.
We write $M_1 \bot M_2$ to mean $\dom{M_1} \cap \dom{M_2} = \varnothing$. 
\end{definition}

\subsection{Operational semantics}

We define a small-step call-by-value operational semantics on configurations of the form $\cfg{e}{\Theta}{\rho}$, where $e$ is a closed expression, $\rho$ is a partial density matrix, and $\Theta$ is a map from qubit references in $e$ to natural number indexes into $\rho$. 
\begin{align*}
    \Theta &::= \emptyset \mid \Theta,q \mapsto i
    \tag{qubit reference maps}
\end{align*}

\begin{definition}
    We say $\rho$ is well-scoped with respect to $\Theta$, written $\Theta \vdash \rho$, if the indices in $\Theta$ all point to valid indices in $\rho$.
\end{definition}

We assume that each unitary $U$ on $k$ qubits has a corresponding unitary operation $\interp{U}_{i_1,\ldots,i_k}$, and transforms the state $\rho$ into $\interp{U}_{i_1,\ldots,i_k} \rho \interp{U}_{i_1,\ldots,i_k}^\dagger$.

As an example, consider applying $H$, a single-qubit unitary gate, to the second qubit in a 3-qubit system.
The configuration would have the form
$\cfg{H(q)}{q \mapsto 1}{\rho}$. In order to appropriately take a step, we will interpret $H$ as a 3-qubit unitary surrounded with identity matrices on either side: $\interp{H}_1 = I_2 \times H \otimes I_2$. The resulting configuration after the step would have the form $\cfg{q}{q \mapsto 1}{(I \otimes H \otimes I) \rho (I \otimes H \otimes I)}$

The step relation, written $\cfg{e}{\Theta}{\rho} \rightarrow \cfg{e'}{\Theta'}{\rho'}$, is shown in \Cref{fig:local-step}. It is made up of $\beta$-reduction rules and contextual rules.
A value is either a bit, a qubit reference, a function, or an expression of the form $!e$. Note that we do not reduce under the bang operator $!e$; it can be thought of as a suspended computation.
\begin{align*}
    v &::= b \mid \REF{q} \mid \lambda x.e \mid \fix{f}{x}{e} \mid !e 
        \tag{values}
\end{align*}

A measurement $\MEAS{\REF{q}}$ is non-deterministic. It will step to either $0$ or $1$, and the probability of each step will be recorded in the trace of the partial density matrix $\rho'$ of the output configuration. In the rule \textsc{meas-B0}, conjugating by the projector $P_0^i = I \otimes \ket{0}\bra{0} \otimes I$ results in the quantum state obtained when measuring $0$; similarly for $P_1^i$.

We write $\cfg{e}{\Theta}{\rho} \rightarrow^\ast \cfg{e'}{\Theta'}{\rho'}$ for the reflective transitive closure of the step relation, i.e.,  $\cfg{e}{\Theta}{\rho}$ can take zero or more steps to reach $\cfg{e'}{\Theta'}{\rho'}$.

\subsection{Typing judgment}

Types include base types for bits and qubits, as well as pairs, linear function types, and duplicable types of the form $!\tau$. $\Gamma$ and $\Delta$ are both typing contexts, with the convention that $\Gamma$ is used for duplicable variables and $\Delta$ is used for non-duplicable variables.
\begin{align*}
    \tau &::= \bit \mid \qubit \mid \tau_1 \otimes \tau_2 \mid \tau_1 \lolli \tau_2 \mid !\tau \tag{types}
    \\
    \Gamma,\Delta &::= \emptyset \mid \Delta,x:\tau 
        \tag{typing contexts}
\end{align*}
The type of a unitary gate is the signature of qubits it acts on:
\begin{align*}
    \text{type}(U) = \begin{cases*}
        \qubit \otimes \qubit & U = \text{CNOT} ~\text{or}~ U = \text{SWAP} \\
        \qubit &\text{otherwise}
    \end{cases*}
\end{align*}

The typing judgment, defined in \cref{fig:local-types}, has the form $\Gamma;\Delta;\Theta \vdash e : \tau$, where $\Gamma$ gives types to the non-linear variables and $\Delta$ gives types to the linear variables. The map $\Theta$ is the same map of qubit references to positions used in the operational semantics. Like linear variables, qubit references are treated linearly, and cannot be duplicated or discarded.

\subsection{Meta-theory}

We prove standard type-safety results in Rocq---type safety via progress and preservation. Below, the title of each lemma or theorem refers to the corresponding lemma in the Rocq proof development.

The proof of progress is fairly straightforward by induction on the typing judgment.

\begin{lemma}[\texttt{Qoreo.Expr.progress}]
    If $\emptyset;\emptyset;\Theta \vdash e : \tau$ and $\Theta \vdash \rho$, then either $e$ is a value, or there exists some configuration $\cfg{e'}{\Theta'}{\rho'}$ such that
    $\cfg{e}{\Theta}{\rho} \rightarrow \cfg{e'}{\Theta'}{\rho'}$.
\end{lemma}

Proving type preservation requires the usual substitution lemmas, both for linear and non-linear variables.

\begin{lemma}~

\begin{enumerate}
\item \texttt{Qoreo.Expr.wt\_subst}: Suppose $\emptyset;\emptyset;\Theta_1 \vdash e : \tau$
    and $\Gamma;\Delta,x:\tau;\Theta_2 \vdash e' : \tau'$ where $\Theta_1 \bot \Theta2$ and $x \notin \domain{\Gamma,\Delta}$.
    Then $\Gamma;\Delta;\Theta_1,\Theta_2 \vdash e'[e/x] : \tau'$. 

    \item \texttt{Qoreo.Expr.wt\_subst\_bang}: If $\emptyset;\emptyset;\emptyset \vdash e : \tau$
    and $\Gamma,x:\tau;\Delta;\Theta \vdash e' : \tau'$, then $\Gamma;\Delta;\Theta \vdash e'[e/x] : \tau'$.
\end{enumerate}
\end{lemma}

Preservation also requires two key lemmas about how qubit reference contexts can change in a step relation depending on which references are used in the expression.

\begin{lemma}[\texttt{Qoreo.Expr.step\_weakening}]
    \label{lemma:e.step_weakening}
    If $\cfg{e}{\Theta_1}{\rho} \rightarrow \cfg{e'}{\Theta_1'}{\rho'}$
    and both $\Theta_1 \bot \Theta_2$
    and $\Theta_1' \bot \Theta_2$, then
    $\cfg{e}{\Theta_1,\Theta_2}{\rho} \rightarrow \cfg{e'}{\Theta_1',\Theta_2}{\rho'}$.
\end{lemma}

\begin{lemma}[\texttt{Qoreo.Expr.step\_inversion}]
    \label{lemma:e.step_inversion}
    Suppose $\Theta_1 \bot \Theta_2$ and
    $\cfg{e}{\Theta_1,\Theta_2}{\rho} \rightarrow \cfg{e'}{\Theta'}{\rho'}$ where $\Theta_1,\Theta_2 \vdash \rho$. Suppose also that $\emptyset;\emptyset;\Theta_1 \vdash e : \tau$, i.e., $e$ only uses qubit references in $\Theta_1$, not those in $\Theta_2$. 
    Then there exists some $\Theta_1'$ disjoint from $\Theta_2$ such that
    $\cfg{e}{\Theta_1}{\rho} \rightarrow \cfg{e'}{\Theta_1'}{\rho'}$ and $\Theta' = \Theta_1',\Theta_2$.
\end{lemma}

With these lemmas, we can prove preservation of the typing relation, as well as preservation of the well-formedness relation $\Theta \vdash \rho$.

\begin{lemma}[\texttt{Qoreo.Expr.preservation}]
    Suppose $\emptyset;\emptyset;\Theta \vdash e : \tau$
    and $\cfg{e}{\Theta}{\rho} \rightarrow \cfg{e'}{\Theta'}{\rho'}$
    where $\Theta \vdash \rho$.
    Then $\emptyset;\emptyset;\Theta' \vdash e' : \tau$ and $\Theta' \vdash \rho'$.
\end{lemma}

Put together, these lead to the following type safety theorem:

\begin{theorem}[\texttt{Qoreo.Expr.safety}]
    \label{thm:e.safety}
    
    Let $\emptyset;\emptyset;\Theta \vdash e : \tau$ and $\Theta \vdash \rho$. If $\cfg{e}{\Theta}{\rho} \rightarrow^\ast \cfg{e'}{\Theta'}{\rho'}$, then either $e'$ is a value, or $\cfg{e'}{\Theta'}{\rho'}$ can take a step.
\end{theorem}

\section{The choreographic language}
\label{sec:choreography}
In this section we define the full Qoreo choreographic language that
integrates the local quantum language from \cref{sec:quantum}. A
choreography is a sequence of instructions which either perform local
computations or communications between actors. We use $A$, $B$ to
range over actor names, which are just unique
identifiers. Instructions can be local to an actor:
$\letin{A.x}{e}{~}$ binds $x$ to the result of a quantum computation
$e$ at $A$ for use in $C$, and similarly $\letin{A.!x}{e}{~}$ binds
the result of the classical computation $e$ at $A$ to $x$, while
$\letin{A.(x_1,x_2)}{e}{~}$ allows $\otimes$-elimination at $A$.  In
addition, instructions include classical communication (send) as
$A.e \rightsquigarrow B.x$ where the result of computation $e$ at $A$
is sent to $B$ as $x$.  Quantum communication and entanglement
generation is supported by the instruction $A.x \leftrightsquigarrow
B.y$, where an entangled pair of qubits is established between $A$ and
$B$ as variables $x$ and $y$.  Entanglement generation is a service
provided by the quantum network~\cite{wehner2018quantum,azuma_quantum_2023}, so we leave its implementation up to the network providers or simulators.
\begin{align*}
    A,B &\in \mathcal{A} \tag{Actor names} \\
    I &::= A.e \rightsquigarrow B.x \mid A.x \leftrightsquigarrow B.y  \\
    &\mid \letin{A.x}{e}{~} \mid \letin{A.(x_1,x_2)}{e}{~} \mid \letin{A.!x}{e}{~} \tag{Instructions}
    \\
    C &::= 0 \mid I; C \tag{Choreographies}
\end{align*}
We note that conditional choreographies are not included in the
syntax.  They are not needed for the examples we've considered and
their omission allows us to avoid issues related to knowledge of
choice \cite{montesi2023introduction}.

In the following sections we present the operational semantics and
typing judgment for choreorgraphies. Both will rely on the introduction
of \emph{actor maps} that rely on and extend Definition \ref{def:mapping}
for choreographies. This is necessary for keeping track of which
values and/or types are localized to which actors.

\begin{definition}
\label{def:actor-mapping}
An \emph{actor map} is a map from actors $A$ to other mappings. We
write $\ol{\Theta}$ to denote actor maps to qubit reference environments $\Theta$, and
similarly $\ol{\Gamma}$ and $\ol{\Delta}$ for maps to typing contexts $\Gamma$ or $\Delta$ respectively.
The notations introduced for other mappings in Definition \ref{def:mapping} also
apply to actor maps, such as concatenation $\ol{M_1},\ol{M_2}$ and domain $\domain{\ol{M}}$.
For any mapping $M_A$, we write $A[M_A]$ for the singleton actor map $A \mapsto M_A$.
We write $A.k \in \ol{M}$ to mean that $k \in \dom{\ol{M}(A)}$.
\end{definition}

\begin{figure*}
\begin{mathpar}
    \inferrule*[right=sendC]
    {\cfg{e}{\Theta_A}{\rho} \rightarrow \cfg{e}{\Theta_A'}{\rho'}}
    {\cfg{A.e \rightsquigarrow B.x; C}{\overline{\Theta},A[\Theta_A]}{\rho} 
        \xrightarrow{A} 
        \cfg{A.e' \rightsquigarrow B.x; C}{\overline{\Theta},A[\Theta_A']}{\rho'}
    }

    \inferrule*[right=sendB]
    {~}
    {
        \cfg{A.!e \rightsquigarrow B.x; C}{\overline{\Theta}}{\rho}
        \xrightarrow{A.!e \rightsquigarrow B}
        \cfg{C[e/B.x]}{\overline{\Theta}}{\rho}
    }
    
    %
    \inferrule*[right=letB]
    {~}
    {
        \cfg{\letin{A.x}{v}{~}\! C}{\overline{\Theta}}{\rho}
        \xrightarrow{A}
        \cfg{C[v/A.x]}{\overline{\Theta}}{\rho}
    }

    \inferrule*[right=entangleB]
    {
        q_1,q_2 ~\text{fresh} \\
        n = \text{num qubits of}~\rho \\
        \overline{\Theta}' = \overline{\Theta},A[\Theta_A,q_1 \mapsto 2^n],B[\Theta_B,q_2 \mapsto 2^n+1]  \\
        \rho' = \rho \otimes \ket{\Phi^+}\bra{\Phi^+}
    }
    {
        \cfg{A.x_1 \leftrightsquigarrow B.x_2; C}{\overline{\Theta},A[\Theta_A],B[\Theta_B] }{\rho}
        \xrightarrow{A \leftrightsquigarrow B}
        \cfg{C[q_1/A.x_1][q_2/B.x_2]}{\overline{\Theta}'}{\rho'}
    }

    \inferrule*[right=delay]
    {
        \cfg{C}{\overline{\Theta}}{\rho} \xrightarrow{\ell} \cfg{C'}{\overline{\Theta}'}{\rho'} \\
        \Actors{I} \cap \Actors{\ell} = \emptyset
    }
    {
        \cfg{I; C}{\overline{\Theta}}{\rho}
        \xrightarrow{\ell}
        \cfg{I; C'}{\overline{\Theta}'}{\rho'}
    }
\end{mathpar}
    \caption{Select operational semantics for choreographies. In the \textsc{send-B} rule, the value being sent must be of the form $!e$ to limit classical communication to only transmitting classical data, never quantum data like qubits. In the \textsc{entangleB} rule, the quantum state is extended with a Bell state of the form $\frac{1}{2}(\ket{00}\bra{00} + \ket{11}\bra{11})$, and fresh qubit references are substituted for $A$'s and $B$'s variables respectively.}
    \label{fig:choreo-step}
\end{figure*}

\subsection{Operational Semantics} 

The operational semantics extends configurations to $\cfg{C}{\overline{\Theta}}{\rho}$,
where $\ol{\Theta}$ tracks qubits allocated to each actor in the choreography. Defining
the step relation for these configurations requires some preliminary details. 
To begin with, we need to be careful about substitutions in choreographies to
ensure they respect process locality, and to ensure that inter-process communication
can only occur via explicit communication instructions. Thus we define process-respecting
substitutions as follows. 
\begin{definition}
  Define $B.e_1[e_2/A.y] = e_1[e_2/y]$ if $A = B$ and $e_1$ otherwise. Then we define
  \emph{process respecting substitution} in instructions, written $I[e/A.y]$, as follows:
    \begin{mathpar}
        (B_1.e' \rightsquigarrow B_2.x)[e/A.y] \defeq (B_1.(B_1.e')[e/A.y] \rightsquigarrow B_2.x)

        (B_1.x \leftrightsquigarrow B_2.y)[e/A.y] \defeq (B_1.x \leftrightsquigarrow B_2.y)     
     
        (\letin{B.x}{e'}{}\!)[e/A.y] \defeq (\letin{B.x}{(B.e')[e/A.y] }{}\!)

        (\letin{B.(x_1,x_2)}{e'}{}\!)[e/A.y] \defeq (\letin{B.(x_1,x_2)}{(B.e')[e/A.y]}{}\!)

        (\letin{B.!x}{e'}{}\!)[e/A.y] \defeq (\letin{B.!x}{(B.e')[e/A.y]}{}\!)             
  \end{mathpar}                        
  Then, writing $\mathit{rebound}(A.y,I)$ to mean that a binding in $I$ shadows
  $A.y$ (e.g., if $I$ is $(\letin{A.y}{e'}{~}\!)$), we inductively define choreographic
  substitution, written $C[e/A.x]$, as follows:
  \begin{mathpar}
        0[e/A.x] = 0

        \inferrule
        {\mathit{rebound}(A.y,I)}
        {(I;C)[e/A.x] = (I[e/A.x];C)}
        
        \inferrule
        {\neg\mathit{rebound}(A.y,I)}
        {(I;C)[e/A.x] = (I[e/A.x];C[e/A.x])}
  \end{mathpar}
\end{definition}

The choreographic step relation is also labeled by the type of action
performed and the location(s) at which the action has taken place:
either a local computation at an actor $A$; a classical communication
between $A$ and $B$ (written $A.v \rightsquigarrow B$), or a classical
communication between $A$ and $B$ (written $A \leftrightsquigarrow
B$).
\begin{align*}
    \ell &::= A \mid A.v \rightsquigarrow B \mid A \leftrightsquigarrow B \tag{labels}
\end{align*}
Note that we consider quantum communication to be symmetric, so the
label $A \leftrightsquigarrow B$ is equivalent to
$B \leftrightsquigarrow A$. We write $\Actors{\ell}$ to denote the actors
occuring in a label, and $\Actors{I}$ to denote the actors occuring in an
instruction. The reason for these labels will become more
clear when we discuss the \textsc{delay} rule that allows computational
steps within the scope of other instructions, as long as $\Actors{\ell}$--
identifying the processes involved in the steps-- are disjoint from
$\Actors{I}$-- identifying the processes whose instructions are ``delayed''.

The step relation then has the form $\cfg{C}{\ol{\Theta}}{\rho} \xrightarrow{\ell}
\cfg{C'}{\ol{\Theta}'}{\rho'}$ and is defined in \Cref{fig:choreo-step}. As we did for
the local step rules, we define contextual $\textsc{C}$ forms and
computational $\textsc{B}$ forms for each instruction. 
Note that contextual rules, where a nested expression reduces in the
actor $A$, are given the label $A$ corresponding to local computation.
The \textsc{delay} rule allows out-of-order execution as long as the
actors involved in the step don't appear in the delayed instruction.

In the \textsc{entangleB} rule, we assume the current quantum state has $n$ qubits. We initialize two fresh qubits in the entangled Bell state
$\ket{\Phi^+} = \frac{1}{\sqrt{2}}(\ket{00}\bra{00}+\ket{11}\bra{11})$ at positions $(n+1)$ and $(n+2)$ in the state respectively. Then, for fresh qubit references $q_1$ and $q_2$, we add the mappings $A.q_1 \mapsto (n+1)$ and $B.q_2 \mapsto (n+2)$ to the qubit mapping references, substituting $q_1$ for $A.x_1$ and $q_2$ for $B.x_2$ in the rest of the choreography.

The multi-step reduction relation $\rightarrow^\ast$ is defined as the
reflexive, transitive closure of $\rightarrow$, ignoring labels. That is:
\begin{mathpar}
  \inferrule*
  {~}
  {\cfg{C}{\overline{\Theta}}{\rho}  \rightarrow^\ast \cfg{C}{\overline{\Theta}}{\rho}}

  \inferrule
  {
    \cfg{C_1}{\overline{\Theta}_1}{\rho_1} \xrightarrow{\ell} \cfg{C_2}{\overline{\Theta}_2}{\rho_2} \\
    \cfg{C_2}{\overline{\Theta}_2}{\rho_2} \rightarrow^\ast \cfg{C_3}{\overline{\Theta}_3}{\rho_3}
  }                                       
  {\cfg{C_1}{\overline{\Theta}_1}{\rho_1} \rightarrow^\ast \cfg{C_3}{\overline{\Theta}_3}{\rho_3}}
\end{mathpar}

\subsection{Typing relation}

The typing judgment for choreographies has the form
$\overline{\Gamma};\overline{\Delta};\overline{\Theta} \vdash C$. The
actor mappings $\overline{\Gamma}$, $\overline{\Delta}$, and
$\overline{\Theta}$ maintain classical, linear, and qubit reference
environments respectively for each process (identified by actor).  We
show selected type derivation rules in \cref{fig:choreography-types}
for the \textsc{send}, \textsc{entangle}, and \textsc{let} cases. We note
especially for the \textsc{entangle} and \textsc{let} cases that linearity
is enforced by disjointness conditions on type environments.  This
ensures no cloning, and no re-use of qubits. The \textsc{send} case
ensures that only results of classical computations are communicated
classically, and that qubits are communicated via entanglement generation.

\begin{figure*}
    \centering
\begin{mathpar}
    \inferrule*[right=send]
    {
        \ol{\Gamma}(A);\Delta_1;\Theta_1 \vdash e : !\tau \\
        \overline{\Gamma},B.x:\tau; \overline{\Delta}, A[\Delta_2]; \overline{\Theta},A[\Theta_2]  \vdash C\\
        \Delta_1 \bot \Delta_2 \\
        \Theta_1 \bot \Theta_2
    }
    {
        \overline{\Gamma} ; \left(\overline{\Delta}, A[\Delta_1,\Delta_2]\right);
        \left(\overline{\Theta}, A[\Theta_1,\Theta_2]\right)
        \vdash A.e \rightsquigarrow B.x; C
    }
    
    \inferrule*[right=entangle]
    {
        \overline{\Gamma}; \overline{\Delta},A.x : \qubit, B.y : \qubit; \overline{\Theta} \vdash C\\
        A.x \not\in \ol{\Gamma},\ol{\Delta} \\
        B.y \not\in \ol{\Gamma},\ol{\Delta}
    }
    {\overline{\Gamma};\overline{\Delta};\overline{\Theta} \vdash A.x \leftrightsquigarrow B.y; C}
    
    \inferrule*[right=let]
    {
        \ol{\Gamma}(A);\Delta_1;\Theta_1 \vdash e : \tau \\
        \overline{\Gamma}; \overline{\Delta}, A[\Delta_2,x:\tau]; \overline{\Theta}, A[\Theta_2] \vdash C\\
        \Delta_1 \bot \Delta_2 \quad
        \Theta_1 \bot \Theta_2 \quad
        A.x \not\in \ol{\Gamma},\ol{\Delta} \\   
    }
    {
        \overline{\Gamma};
        \left(\overline{\Delta}, A[\Delta_1,\Delta_2]\right);
        \left(\overline{\Theta}, A[\Theta_1,\Theta_2]\right)
        \vdash \letin{A.x}{e}{C}
    }
\end{mathpar}
    \caption{Select typing rules for the choreographic language.}
    \label{fig:choreography-types}
\end{figure*}

\subsection{Meta-theory}

Now we extend our type safety results to the choreographic language. Since evaluation of 
choreographies incorporates evaluation of local computations, type safety for choreographies 
requires type safety results for local computations. Type preservation for top-level
choreographies can be stated as follows, where $\overline{\Theta} \vdash \rho$ means
$\overline{\Theta}(A) \vdash \rho$ for all $A \in \dom{\ol{\Theta}}$.
\begin{lemma}
    \label{lemma:C.preservation-hl}
    If $\emptyset;\emptyset;\overline{\Theta} \vdash C$ and
    $\cfg{C}{\overline{\Theta}}{\rho} \xrightarrow{\ell} \cfg{C'}{\overline{\Theta}'}{\rho'}$ where
    $\overline{\Theta} \vdash \rho$, then $\emptyset;\emptyset;\overline{\Theta}' \vdash C'$ and
    $\overline{\Theta}' \vdash \rho'$.
\end{lemma}

A main challenge in proving choreographic type safety is
the \textsc{delay} rule that allows evaluation of "subchoreographies"
of a choreography $C$. This means that type preservation in particular
needs to be defined with respect to choreographies with free
variables for actors not involved with the reduction step. The
latter can always be determined from the label on the reduction step.
Thus we demonstrate the following lemma, from which lemma
\label{lemma:C.preservation-hl} follows as a special case.
This statement allows an induction on typing derivations
and inversion on the reduction step.
\begin{lemma}[\texttt{Qoreo.Choreography.preservation}]
    \label{lemma:C.preservation}
    If $\overline{\Gamma};\overline{\Delta};\overline{\Theta} \vdash C$ where
    $\overline{\Gamma}(A) = \emptyset$ and $\overline{\Delta}(A) = \emptyset$
    for all $A \in  \Actors{\ell}$ and 
    $\cfg{C}{\overline{\Theta}}{\rho} \xrightarrow{\ell} \cfg{C'}{\overline{\Theta}'}{\rho'}$ where
    $\overline{\Theta} \vdash \rho$, then
    $\overline{\Gamma};\overline{\Delta};\overline{\Theta}' \vdash C'$ and
    $\overline{\Theta}' \vdash \rho'$.
\end{lemma}

The proof of lemma \ref{lemma:C.preservation} proceeds in a ``traditional manner''
in that it incorporates substitution lemmas to apply in cases where reduction
involves substitutions. Since substitutions can be either classical,
as in \textsc{sendB}, or linear, as in \textsc{entangleB}, we need different
lemmas for classical and linear substitution. Classical substitution is
stated as follows. 
\begin{lemma}[\texttt{Qoreo.Choreography.wt\_subst\_bang}]
    \label{lemma:C.wt_subst_bang}
    $\overline{\Gamma},A.x : \tau;\overline{\Delta};\overline{\Theta} \vdash C$ and
    $\emptyset,\emptyset,\emptyset \vdash e : \tau$ then
    $\overline{\Gamma};\overline{\Delta};\overline{\Theta} \vdash C[e/A.x]$.
\end{lemma}

Linear substitution-- that is, a substitution $C[v/A.x]$ when $v$ is a
linear value-- provides insights into type system features. In
contrast to classical substitution, it imposes linearity requirements
on $A.x$, namely that it only occurs once in the linear
environment. The variable $A.x$ should also not occur in the classical
environment to avoid conflicts in the typing of $A.x$. Also, the Qref
ids in $v$ should not overlap with the Qref ids in $C$-- this is
important to disallow cloning.
\begin{lemma}[\texttt{Qoreo.Choreography.wt\_subst\_lin}]
    \label{lemma:C.wt_subst_lin}
    $\overline{\Gamma};\overline{\Delta},A.x : \tau;\overline{\Theta},A[\Theta_2] \vdash C$ and
    $\emptyset,\emptyset,\Theta_1 \vdash e : \tau$
    and $\Theta_1\bot\Theta_2$ and $A.x\not\in \ol{\Gamma},\ol{\Delta}$, then
    $\overline{\Gamma};\overline{\Delta};\overline{\Theta},A[\Theta_1,\Theta_2] \vdash C[e/A.x]$.
\end{lemma}
In addition to these lemmas, we also needed establish choreographic analogues of the
quantum step weakening and inversion lemmas (\ref{lemma:e.step_weakening} and
\ref{lemma:e.step_inversion}). These are mainly used in the delay cases,
when the continuation choreography steps before local computation in the
current instruction. 

Progress is stated as follows. It shows the type system disallows bad things from
happening-- cloning qubits are prevented by linearity, only classical values
are bound to classical variables, tensors are deconstructed and used linearly, etc. 
\begin{lemma}[\texttt{Qoreo.Choreography.progress}]
    \label{lemma:C.progress}
    Suppose $\emptyset;\emptyset;\Theta \vdash C$ and $\Theta \vdash \rho$. Then either $C$
    is the empty choreography, or there exists a configuration $\cfg{C'}{\Theta}'{\rho'}$ and
    label $\ell$ such that $\cfg{C}{\Theta}{\rho} \xrightarrow{\ell} \cfg{C'}{\Theta}'{\rho'}$
\end{lemma}
Putting the above together, we can state type safety as follows.
The result is proved by induction on $\rightarrow^\ast$ and as a consequence of lemmas \ref{lemma:C.preservation}
and \ref{lemma:C.progress}. In short, lemma \ref{lemma:C.preservation} guarantees preservation
of typing and well-scopedness of configurations through multiple steps, and lemma
\ref{lemma:C.progress} guarantees progress at every reduction step of the given well-typed choreography. 
\begin{theorem}[\texttt{Qoreo.Choreography.safety}]
    \label{thm:C.safety}
    Let $\emptyset;\emptyset;\Theta \vdash C$ and
    $\Theta \vdash \rho$.  Then whenever
    $\cfg{C}{\Theta}{\rho} \rightarrow^\ast \cfg{C'}{\Theta'}{\rho'}$,
    the configuration $\cfg{C'}{\Theta'}{\rho'}$ is not stuck.
\end{theorem}

\section{The process language}
\label{sec:network}
EPP for Qoreo can be established with a relatively simple process language. A process describes a local computation run running independently by an actor, with the addition of (one-sided) communication patterns. For classical communication, these include \text{send} and \text{receive}; for quantum communication they include the ability to establish an EPR pair with another actor $B$, which we write $\establishwith{x}{B}{P}$.
\begin{align*}
    P &::= 0 \mid \sendin{v}{B}{P} \mid \letin{x}{e}{P} \mid \letin{!x}{e}{P} \mid \letin{(x_1,x_2)}{e}{P} \\
    & \mid \receivefrom{B}{x}{P} \mid \establishwith{x}{B}{P} \tag{processes}
\end{align*}
A network is a map from actors to processes, enabling communication between two processes with matching communication actions. Note that we consider two networks equivalent up to reordering.
\begin{align*}
    N &::= A_1[P_1] \parallel \ldots \parallel A_n[P_n]
        \tag{networks}
\end{align*}

\subsection{Operational semantics}

Individual processes can step independently-- i.e., perform local computation-- if not blocked by a communication pattern. This relation is written $\cfg{P}{\Theta}{\rho} \rightarrow \cfg{P'}{\Theta'}{\rho'}$. The rules include the usual contextual and $\beta$-reduction rules for non-communication patterns. For example:\footnote{For the full set of reduction rules for processs, see the Rocq development.}
\begin{mathpar} \small
    \inferrule*[right=letPC]
    {
        \cfg{e}{\Theta}{\rho} \rightarrow \cfg{e'}{\Theta'}{\rho'}
    }
    {
        \cfg{\letin{x}{e}{P}}{\Theta}{\rho}
        \rightarrow
        \cfg{\letin{x}{e'}{P}}{\Theta'}{\rho'}
    }
    
    \inferrule*[right=letPB]
    {
        ~
    }
    {
        \cfg{\letin{x}{v}{P}}{\Theta}{\rho}
        \rightarrow
        \cfg{P[v/x]}{\Theta}{\rho}
    }
\end{mathpar}

The operational semantics of a network has a similar structure to that of choreographies, with a judgment $\cfg{N}{\overline{\Theta}}{\rho} \xrightarrow{\ell} \cfg{N'}{\overline{\Theta}'}{\rho'}$, as shown in \cref{fig:semantics-network}. It allows for local computations as well as classical and quantum communication actions. In classical communication, the send and receive components of the communication must by synchronized. In quantum communication, entanglement generation must by synchronized between two actors.

\begin{figure}
    \centering \footnotesize
\begin{mathpar}
    \inferrule*[right=local]
    {
        \cfg{P_A}{\Theta_A}{\rho} \rightarrow \cfg{P_A'}{\Theta_A'}{\rho'}
    }
    {
        \cfg{A[P] \parallel N}{\left(A[\Theta], \overline{\Theta}\right)}{\rho}
        \xrightarrow{A}
        \cfg{A[P'] \parallel N}{\left(A[\Theta'], \overline{\Theta}\right)}{\rho'}
    }

    \inferrule*[right=com-C]
    {
        A \neq B
    }
    {
        \cfg{A[\sendin{!e}{B}{P_A}] \parallel B[\receivefrom{A}{x}{P_B}] \parallel N}{\overline{\Theta}}{\rho}
        \xrightarrow{A.e \rightsquigarrow B}
        \cfg{A[P_A] \parallel B[P_B[e/x]] \parallel N}
            {\overline{\Theta}}{\rho}
    }

    \inferrule*[right=com-Q]
    {
        A \neq B \\
        q_1,q_2 ~\text{fresh} \\
        n = \text{num qubits of}~\rho \\
        \overline{\Theta}' = \overline{\Theta},A.q_a \mapsto n,B.q_B \mapsto n+1 \\
        \rho' = \rho \otimes \ket{\Phi^+}\bra{\Phi^+} \\
        N' = A[P_A[q_A/x]] \parallel B[P_B[q_b/y]] \parallel N
    }
    {
        \cfg{A[\establishwith{x}{B}{P_A}] \parallel B[\establishwith{y}{A}{P_B}] \parallel N}
            {\overline{\Theta}}{\rho} 
        \xrightarrow{A \leftrightsquigarrow B}
        \cfg{N'}{\overline{\Theta}'}{\rho'}
    }
\end{mathpar}
    \caption{Operational semantics rules for networks. Substitution in processes $P[e/x]$ is defined in the usual way.}
    \label{fig:semantics-network}
\end{figure}

\subsection{Type system}

The process language does not have a type system itself; it will inherit type-safety as well as deadlock-safety from the soundness and correctness of endpoint projection.

\section{Endpoint Projection}
\label{sec:epp}
Endpoint projection maps choreographies onto the set of processes that run on individual nodes. This is accomplished by ``filtering'' choreographies to extract, for each actor $A$, the process (set of instructions) they are involved in-- that is, their local computations and communications where they are either sender or receiver. Given a choreography $C$ and an actor $A$, we write $\interp{C}_A$ for the corresponding process run by $A$:
\begin{align*}
    \interp{B_1.e \rightsquigarrow B_2.x; C}_A
        &= \begin{cases}
            \sendin{e}{B_2}{\interp{C}_A} & A = B_1 \\
            \receivefrom{B_1}{x}{\interp{C}_A} & A = B_2 \neq B_1 \\
            \interp{C}_A &\text{otherwise}
        \end{cases} \\
    \interp{B_1.x_1 \leftrightsquigarrow B_2.x_2; C}_A
        &= \begin{cases}
            \establishwith{x_1}{B_2}{\interp{C}_A} & A = B_1 \\
            \establishwith{x_2}{B_1}{\interp{C}_A} & A = B_2 \neq B_1 \\
            \interp{C}_A &\text{otherwise}
        \end{cases} \\
    \interp{\letin{B.x}{e}{C}}_A
        &= \begin{cases}
            \letin{x}{e}{\interp{C}_A} & A = B \\
            \interp{C}_A & \text{otherwise}
        \end{cases} \\
    \interp{\letin{B.!x}{e}{C}}_A
        &= \begin{cases}
            \letin{!x}{e}{\interp{C}_A} & A = B \\
            \interp{C}_A & \text{otherwise}
        \end{cases} \\
    \interp{\letin{B.(x_1,x_2)}{e}{C}}_A
        &= \begin{cases}
            \letin{(x_1,x_2)}{e}{\interp{C}_A} & A = B \\
            \interp{C}_A & \text{otherwise}
            \end{cases}\\
     \interp{0}_A
        &= 0
\end{align*}

Notice that not all choreographies are projectable; specifically, if a choreography has a communication primitive from one actor to itself, it will not be projected correctly. We say that choreographies without self-send are \emph{well-formed}. In our EPP correctness results we will restrict consideration to well-formed choreographies. However, typing enforces well-formedness so this will be implied by our main network safety results. Note that these require the assumption of typability, since projection alone does not guarantee semantic safety (not getting stuck). We begin with lemmas establishing these basic properties. 
\begin{lemma} Both of the following hold:
  \label{lemma:wtwf}    
    \begin{enumerate}
        \item (\texttt{Qoreo.Choreography.WellTyped\_WellFormed}): Well-typed choreographies are well-formed
        \item (\texttt{Qoreo.Network.EPP\_projectable}) If $C$ is well-formed, then $C$ is projectable; i.e. for every $A$, there exists some $P$ such that $\interp{C}_A = P$.
    \end{enumerate}
\end{lemma}
We construct the network projected from a choreography $C$, denoted $\interp{C}$, by process filtering each actor $A_1,\ldots,A_n$ occuring in $C$: 
\begin{align*}
    \interp{C} &= A_1[\interp{C}_{A_1}] \parallel \cdots \parallel A_n[\interp{C_n}_{A_n}]
\end{align*}

Now we demonstrate correctness of EPP by proving both soundness and completeness. We prove that EPP is sound and complete with respect to the semantics of choreographies. Soundness says that if a choreography takes a step, it's projection takes a step to a projection of its redex.
\begin{theorem}[\texttt{Qoreo.Network.soundness}]
    \label{thm:epp-sound}
    Suppose $\cfg{C}{\Theta}{\rho} \xrightarrow{\ell} \cfg{C'}{\Theta'}{\rho'}$ and $C$ is well-formed. Then $\cfg{\interp{C}}{\Theta}{\rho} \xrightarrow{\ell} \cfg{\interp{C'}}{\Theta'}{\rho'}$.
\end{theorem}

Completeness says that if a network takes a step, and it is the projection of a choreography, the projected choreography takes a corresponding step. To prove this, we needed to reason about the behavior of EPP. First, we prove the following fact about process-respecting substitution in choreographies.
\begin{lemma}[\texttt{Qoreo.Network.EPP\_subst\_iff}]
Both of the following hold:
\begin{enumerate}
    \item $\interp{C[v/A.x]}_A = PA$ if and only if $\interp{C}_A = PA_0$ and $PA = PA_0[v/x]$.
    \item If $D\neq A$ then $\interp{C[v/A.x]}_D = PD$ if and only if $\interp{C}_D=PD$.
\end{enumerate}   
\end{lemma}

The following lemma supports a basic proof strategy for completeness. Intuitively, it allows us to deconstruct a given choreography into an instruction and continuation, and to find the location of its projection in the network.

\begin{lemma}[\texttt{Qoreo.Network.EPP\_N\_cons\_inversion}]
    If $\interp{I;C} = N$ then there exists some network $N'$ such that $\interp{C} = N'$. Specifically, $N'$ is defined so that whenever $N(A)=PA$ and $A \in \Actors{I}$, then $N'(A)=\interp{I;C}_A$.
\end{lemma}





With the above results, we are able to prove completeness of endpoint projection.
\begin{theorem}[\texttt{Qoreo.Network.completeness}]
    \label{thm:epp-complete}
    If $\cfg{\interp{C}}{\Theta}{\rho} \xrightarrow{\ell} \cfg{N'}{\Theta'}{\rho'}$ for a well-formed choreography $C$, then there exists some $C'$ such that 
    $\cfg{C}{\Theta}{\rho} \xrightarrow{\ell} \cfg{C'}{\Theta'}{\rho'}$ and $\interp{C'}=N'$.
\end{theorem}
\begin{proof}[Proof sketch]
    Case analysis on $\ell$, then induction on the network step relation.
\end{proof}

As a consequence of our formal development, we can assert a network safety result-- that is, we can conclude that networks projected from well-typed choreographies do not go wrong, and particularly do not have communication deadlocks. Lemma \ref{lemma:wtwf} guarantees projectability, while correctness as established in Theorems \ref{thm:epp-sound} and \ref{thm:epp-complete} together with type safety in Qoreo as established in Theorem \ref{thm:C.safety} obtain the result. 
\begin{corollary}[\texttt{Qoreo.Network.safety}]
    \label{cor:net.safety}
    Suppose $\emptyset;\emptyset;\Theta \vdash C$ and $\Theta \vdash \rho$. If $\cfg{\interp{C}}{\Theta}{\rho} \rightarrow^\ast \cfg{N'}{\Theta'}{\rho'}$, then either $N'$ is the terminal network---$N'(A) = 0$ for all actors $A$---or $N'$ can take a step.
\end{corollary}

\section{Implementation and extraction to NetQASM}
\label{sec:extraction}


\Cref{fig:rocq} shows the definition of expressions and choreographies in Rocq. Expressions and instructions are defined as a deep embedding via inductive datatypes for each abstract syntax tree. Choreographies themselves are lists of instructions.

\begin{figure}
\begin{subfigure}[b]{0.39\textwidth}
{\footnotesize\begin{verbatim}
Module Expr.
  Inductive t :=
  | Var : Var.t -> t
  | LetIn : Var.t -> t -> t -> t
  | QRef : Var.t -> t
  | Meas : t -> t
  | New : t -> t
  | Unitary : unitary -> t -> t
  | ...
End Expr
\end{verbatim}}
\end{subfigure}
\quad
\begin{subfigure}[b]{0.4\textwidth}
{\footnotesize\begin{verbatim}
Module Insn.
  Inductive t :=
   | Send (A:Actor.t) (e:Expr.t) (B:Actor.t) (x:Var.t)
   | Entangle : Actor.t -> Var.t -> Actor.t -> Var.t -> t
  | ...
End Insn.
...
Module Choreography.
  Definition t := list Insn.t.
End Choreography.
\end{verbatim}}
\end{subfigure}
    \caption{Rocq data structures representing Qoreo ASTs. We use the OCaml/Rocq convention of defining modules for major datatypes (\texttt{Expr}, \texttt{Insn}, \texttt{Choreography}) with \texttt{t} as its type.}
    \label{fig:rocq}
\end{figure}

We use a named representation of variables, represented internally as natural numbers.
For example, the quantum teleportation example from the introduction could be written explicitly as:
{\small\begin{verbatim}
Definition teleport Alice Bob q : Choreography.t :=
  [Insn.Entangle Alice 2 Bob 3;
  Insn.LetPair Alice 4 5 (Unitary CNOT (Pair (Expr.Var q) (Expr.Var 2)));
  Insn.Let Alice 6 (Unitary H (Expr.Var 4));
  Insn.Let Alice 7 (Meas (Expr.Var 6));
  Insn.Let Alice 8 (Meas (Expr.Var 5));
  Insn.Send Alice (Expr.Var 7) Bob 9;
  Insn.Send Alice (Expr.Var 8) Bob 10;
  Insn.Let Bob 11 (If (Expr.Var 10) (Unitary Z (Expr.Var 3)) (Expr.Var 3));
  Insn.Let Bob 12 (If (Expr.Var 9) (Unitary X (Expr.Var 11)) (Expr.Var 11))].
\end{verbatim}}
To avoid this explicit variable manipulation, we introduce a small domain-specific language to construct choreographies using automatic variable renaming using a state monad. The state monad, which we write \texttt{Qoreo A}, tracks fresh variables and constructively builds a choreography with each step. 
We can extract the choreography from a monadic computation using the following function:
\small{\begin{verbatim}
Definition run {T} (m : Qoreo T) (init_var : Var.t) : Network.Choreography.t * T :=
  let (s, v) := m {| max_var := init_var+1; chor := [] |} in
  (chor s, v).
\end{verbatim}}
With the addition of some custom \texttt{Notation} syntax, we can express the teleportation program as follows:
{\small\begin{verbatim}
Definition teleport Alice Bob q C : Qoreo Var.t :=
      (* Alice and Bob establish an entangled pair of qubits. *)
      do ( a , b ) ← get_entangled_pair Alice Bob ;;
      
      (* Alice performs some local operations and obtains classical bits x and z. *)
      do (q,a) ← Alice [-- Unitary CNOT (Pair q a) -] ;;
      do q     ← Alice [- Unitary H q -] ;;
      do x ← Alice [- Meas q -];;
      do z ← Alice [- Meas a -];;

      (* Alice sends x and z to Bob. *)
      do x ← send Alice x Bob ;;
      do z ← send Alice z Bob ;;

      (* Bob uses x and z to update his qubit. *)
      do b ← Bob [- If z (Unitary Z b) b -];;
      do b ← Bob [- If x (Unitary X b) b -];;
      C b
\end{verbatim}}

Notice that \texttt{teleport} takes a continuation $C$ as input, because otherwise it would terminate with a free variable, \texttt{b}, still in scope. This would not be well-typed according to our typing judgment, as the empty choreography is only typeable when there are no linear variables in the context. 
This is reflected in the typing rule derivable for \texttt{teleport}, which must be in the context of another choreography \texttt{C} that uses \texttt{b}. 

\begin{lemma}[\texttt{QoreoExamples.Teleportation.WellTyped\_teleport}]
    Suppose $qA$ is fresh and $\emptyset; Bob[qB:\qubit]; \emptyset \vdash C(qB)$.
    Then $\emptyset; Alice[qA:\qubit]; \emptyset \vdash \texttt{tleeport}~Alice~Bob~qA~C$.
\end{lemma}

\subsection{Extraction to NetQASM}

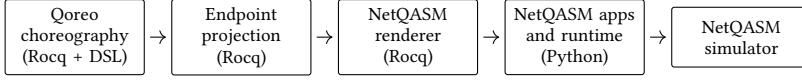
\begin{figure}
\centering
\begin{tikzpicture}[
  component/.style={
    draw,
    rounded corners=1pt,
    align=center,
    text width=1.7cm,
    minimum height=2.5em,
    inner xsep=2pt,
    inner ysep=3pt,
    font=\scriptsize
  },
  flow/.style={->, shorten >=2pt, shorten <=2pt}
]
\node[component] (qoreo) {Qoreo\\choreography\\(Rocq + DSL)};
\node[component, right=0.35cm of qoreo] (epp) {Endpoint\\projection\\(Rocq)};
\node[component, right=0.35cm of epp] (renderer) {NetQASM\\renderer\\(Rocq)};
\node[component, right=0.35cm of renderer] (apps) {NetQASM apps\\and runtime\\(Python)};
\node[component, right=0.35cm of apps] (netqasm) {NetQASM simulator};
\draw[flow] (qoreo) -- (epp);
\draw[flow] (epp) -- (renderer);
\draw[flow] (renderer) -- (apps);
\draw[flow] (apps) -- (netqasm);
\end{tikzpicture}
\caption{Extraction pipeline from a global Qoreo choreography to simulation using NetQASM.}
\label{fig:qoreo-extraction}
\end{figure}

Qoreo comes with an extraction path from verified choreographies in Rocq to executable NetQASM applications~\citep{dahlberg_netqasm_2022}. The implementation follows the structure of the formal development: programmers write a single global choreography, the framework projects this choreography to one local process per participant, and the projected processes are rendered as NetQASM SDK programs that can be run by the standard NetQASM application infrastructure. This extraction path is illustrated in \cref{fig:qoreo-extraction}.


The final verified-to-executable step is provided by the NetQASM renderer, which renders Qoreo processes as Python source code for a small Qoreo runtime API. The runtime API in turn implements operations by calling the NetQASM API. The renderer produces one application file per actor, including the actor's local instructions and the peer information needed to construct its classical and entanglement connections. Expressions in the choreographic language are translated into their corresponding runtime operations.

The renderer uses Rocq's native extraction mechanism to first convert Qoreo processes to OCaml data structures. 
Executing the extracted OCaml code runs endpoint projection and generates the Python applications, one file per participant, together with the shared runtime module.

The generated Python programs use the Qoreo NetQASM runtime as a thin layer over the NetQASM SDK. This runtime constructs classical sockets and EPR sockets from the peer sets computed during rendering, manages the NetQASM connection, and exposes the operations used by the generated code. As a result, the extraction framework produces ordinary NetQASM SDK applications while preserving a direct line from the global choreographic source, through endpoint projection, to the local code executed by each actor.

\section{Examples}
\label{sec:examples}

In \Cref{sec:introduction,sec:extraction} we explored quantum teleportation in Qoreo, and in this section we will illustrate other examples of distributed quantum protocols in Qoreo.

\subsection{Distributed controlled gates}
\label{sec:example:controlled}

Consider the protocol in \Cref{fig:DCU-protocol}. It shows an implementation of a distributed $2$-qubit controlled unitary gate across two actors. The protocol is made up of three main stages. In stage 1, Alice and Bob establish an entangled pair of qubits $a$ and $b$. Alice entangles her pre-existing qubit with $a$ and sends a correction to Bob. In stage 2, Bob performs the controlled-unitary locally on his own qubits. In stage 3, Bob sends another correction back to Alice, which effectively uncomputes the effect of the entanglement on Alice's qubits.\footnotemark

\footnotetext{
Notice that this protocol is an example of why entanglement, and not teleportation, is a good primitive for quantum communication. We could have implemented a distributed gate by teleporting Alice's control qubit to Bob, performing the controlled gate, and then teleporting the control qubit back to Alice. However, that protocol would use two entangled states and four classical bits of communication, whereas this protocol uses only one entangled state and two classical communication bits.
}

\begin{figure}
    \begin{quantikz}
        \lstick{q~~~~~~}  
        &  \gategroup[2,steps=7,style={dashed,rounded corners,fill=blue!20, inner xsep=1pt}, background,label style={label position=left,anchor=west,xshift=-1cm,yshift=-0.75cm}]{Alice}
        &             & \ctrl{1}                 & & &   &  \gate{Z} &  \\
        &
        \setwiretype{n}  & \gate[4]{\text{entanglement}} \gateoutput{$a$}
            & \targ{} \setwiretype{q}  &  \gate{\text{Meas}} \wire[d][3]{c}
            \\
        \setwiretype{n}  &                           
            &                          &          &                                   & &
            &                          &  &                      \\
        \setwiretype{n}  & &
            &                          &          &                                   & 
                             \\
        & \gategroup[2,steps=7,style={dashed,rounded corners,fill=red!20, inner xsep=1pt}, background,label style={label position=left,anchor=west,xshift=-1cm,yshift=0cm}]{Bob}
        \setwiretype{n}  &  \gateoutput{$b$}
            & \setwiretype{q}          &          
            \gate{X} \slice{} & \gate[2]{CU} \slice{} & \gate{H} & \gate{\text{Meas}} \wire[u][4]{c} \\
        \lstick{tgt~~~~~~} & & & & & & & & 
    \end{quantikz}
    \caption{Distributed protocol for implementing a 2-qubit controlled-unitary~\cite{sarvaghad2021general}, broken up into three stages.}
    \label{fig:DCU-protocol}
\end{figure}
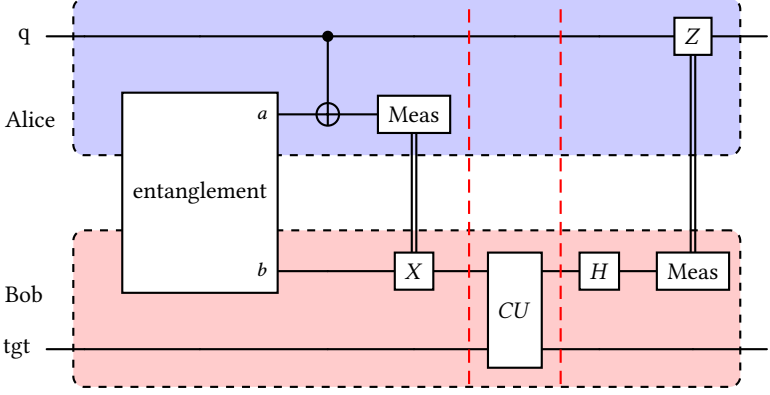

The protocol above is adapted from \citet{sarvaghad2021general}, which presents a generalization to an $n$-qubit controlled gate distributed across $n+1$ actors. This generalized protocol is harder to express as a circuit diagram, which is why we only illustrated the 2-qubit version. However, their $n$-qubit protocol can be easily expressed programatically as a function in Qoreo, shown in \cref{fig:distributed-U-qoreo}.

\begin{figure}
    \centering
\footnotesize
\begin{verbatim}
Definition distributed_n_ctrl_U (CU : list Var.t -> Var.t -> Qoreo (list Var.t * Var.t)) 
                                (As : list Actor.t) (qs : list Var.t)
                                (B : Actor.t) (tgt : Var.t)
                              : Qoreo (list Var.t * Var.t) :=

    let stage1 := fun (A : Actor.t) (q : Var.t) => 
        do (a,b) ← get_entangled_pair A B ;;
        do (q,a) ← A [-- Unitary CNOT (Pair q a) -] ;;
        do x     ← send A (Meas a) B ;;
        do b     ← B [- If x (Unitary X b) b -];;
        ret (q, b) : Qoreo (Var.t * Var.t) in

    let stage3 := fun A q b =>
        do z ← send B (Meas (Unitary H b)) A ;;
        do q ← B [- If z (Unitary Z q) q -];;
        ret q : Qoreo Var.t in

    (* Stage 1: As performs local operations  *)
    do qs_bs_lis ← fmap As qs stage1 ;;
    (* split up qs_bs_list : list (Var.t * Var.t) into a pair of lists *)
    let (qs, bs) := List.split qs_bs_list in

    (* Stage 2: B applies ctrl-U locally using bs as controls *)
    do (qs, tgt) ← CU bs tgt ;;

    (* Stage 3: B sends corrections back to As *)
    do qs ← fmap2 As qs bs stage3 ;;
    
    ret (qs, tgt).
\end{verbatim}
    \caption{An implementation of an $n$-controlled gate distributed across $n$ actors. The function \texttt{distributed\_n\_CU} takes as input a list of actors $A$ with qubit variables $qs$, and a single target actor $B$ with qubit $tgt$. The goal of the protocol is to apply the input controlled unitary $CU$ to the controls $qs$ and target $tgt$. The helper functions \texttt{fmap} and \texttt{fmap2} fold over the input lists and concatenate the results together, e.g. \texttt{stage1} returns a pair of variables and \texttt{fmap As qs stage1} returns a list of pairs of variables, which we split up into a pair of lists using \texttt{List.split}. Note that although \texttt{fmap}/\texttt{fmap2} iterate sequentially through the actors in $As$, the \texttt{Delay} rule of choreographies ensures that execution of each stage can be interleaved depending on the nondetermistic semantics of each actor.}
    \label{fig:distributed-U-qoreo}
\end{figure}

\subsection{Quantum Key Distribution / B92}

In quantum key distribution (QKD), two parties (Alice and Bob) establish a secret key using a quantum channel and a classical channel. The quantum channel is \textit{insecure}: an eavesdropper (Eve) may measure, resend, or otherwise manipulate the entangled pairs it generates. The classical channel is \textit{authenticated} but not \textit{confidential}: Eve may read the messages that pass between Bob and Alice but may not modify them.

The B92 protocol runs in a predetermined number of rounds. In each round, Alice and Bob try to establish a single bit of the secret key. Alice picks a random classical bit and encodes it as a qubit using one of two non-orthogonal bases. Bob randomly chooses one of the bases in which to measure the qubit. If the result is 1, the measurement is said to be \textit{conclusive}. So, the choice of bases is what determines the shared bit; and the measurement serves as an indicator that their choices agree. The probability that Alice and Bob can use the bit is only 1/4, and the number of rounds in the protocol should account for this.

A Rocq definition for a single round of the protocol is shown in \cref{fig:b92-qoreo}. The definition uses quantum randomness for the step involving the creation of a random classical bit, instead of a pseudorandom number generator, simply because the machinery for that technique is readily available here. Note also the use of the teleportation abstraction to meet the way in which the B92 setup is usually described: in terms of \textit{sending} a qubit across the channel.

\begin{figure}
    \centering
    \footnotesize
    \begin{verbatim}
Definition b92_round (Alice Bob : Actor.t) : Qoreo ((Var.t * Var.t) * (Var.t * Var.t)) :=

  (* Alice randomly picks her key bit by preparing |+⟩ and measuring. *)
  do coin_a ← Alice [- Unitary H (New (Bit false)) -] ;;
  do a ← Alice [- Meas coin_a -] ;;

  (* Alice prepares her transmission qubit: |0⟩ if a=0, |+⟩ if a=1. *)
  do q ← Alice [- New (Bit false) -] ;;
  do q ← Alice [- If a (Unitary H q) q -] ;;

  (* Alice sends the qubit to Bob over the quantum channel. *)
  do q_b ← teleport Alice Bob q ;;

  (* Bob randomly picks his measurement basis. *)
  do coin_b ← Bob [- Unitary H (New (Bit false)) -] ;;
  do b ← Bob [- Meas coin_b -] ;;

  (* Bob rotates into his chosen basis by applying H when b=1. *)
  do q_b ← Bob [- If b (Unitary H q_b) q_b -] ;;

  (* Bob measures in his chosen basis.  A result of 1 means the round is
  conclusive: both Alice and Bob know they share the same key bit. *)
  do r ← Bob [- Meas q_b -] ;;

  (* Bob announces the round result to Alice over the classical channel. *)
  do r_recv ← send Bob r Alice ;;

  ret ((a, r_recv), (b, r)).
    \end{verbatim}
    \caption{B92 Algorithm}
    \label{fig:b92-qoreo}
\end{figure}

The complete example in the repository uses \texttt{b92\_round} to model a protocol with three rounds. In its compiled form, the protocol writes Alice's and Bob's keys to their NetQASM logs. For empirical testing, we used a script that runs the NetQASM simulation multiple times, each time analyzing the logs to ensure that the keys match. In a real world scenario, Alice and Bob would also check to see if the number of usable bits is significantly different from 1/4 the number of rounds, as this would implicate Eve in trying to learn the key.

\section{Related Work}
\label{sec:related}

\subsection{Classical choreographies and process languages}

Choreographies and endpoint projection formalize an intuitive idea---to write distributed or networked processes as a global program~\cite{montesi2023introduction}. The introduction of endpoint projection in particular enabled the idea to become a useful programming paradigm, and researchers have extended choreographies with a variety of techniques including procedures~\cite{ProceduralChoreographicProgramming}; higher-order functions~\cite{HasChorFunctionalChoreographic}; knowledge of choice~\cite{montesi2023introduction}; and object-oriented programming~\cite{giallorenzoChoralObjectorientedChoreographic2024}, to name a few.


Choreographies specify what every actor does at every point in a protocol---both a strength and a weakness of the system. As a programmatic style it is excellent: a single choreographic term captures the entire protocol in a readable form. But as a specification it leaves no room for actor autonomy, since every participant's behavior is fully determined. This makes it difficult to reason about security properties or multi-party computation, where different actors act independently without complete knowledge of each other's behavior.

\paragraph{Multi-party session types}
This is addressed in part by multi-party session types~\cite{honda2008multiparty}, which grew out of the domain of binary session types~\cite{HondaVasconcelosKubo1998} and process calculi~\cite{caires2010session}. In multiparty session types, a global type is used to specify the intended communications of the systems, without specifying the values being transmitted. The global type is projected onto local binary session types implemented by independent processes. In some cases, multiparty session types can be thought of as a type system for choreographies~\cite{carbone2018multiparty}.

Both choreographies and multi-party session types present a global view of a protocol, which typically imposes some additional limitations. For example, projection typically requires knowing all the actors in the system ahead of time, which can make it difficult to design choreographic systems in a compositional style, where protocols can be written (and even projected) independently and later combined. Some work in the classical setting has addressed this challenge---for example, \citet{hirsch2022pirouette} support compositionality through higher-order choreographic functions---but these approaches add a certain amount of complexity. 

A related limitation is that choreographies are not always suited to expressing protocols where new actors can be dynamically spawned during execution. For near-term quantum networks, where hardware resources are extremely constrained, this may not be a practical concern---the set of participating nodes is typically fixed at deployment time. In the longer term, however, scalable quantum networks will likely need to support dynamic participation, and extending choreographies to handle this case is a non-trivial open problem.

\paragraph{Binary session types.}
In contrast, binary session types~\cite{HondaVasconcelosKubo1998} describe the behavior of each endpoint independently, and are thus much more composable. As long as each actor implements the prescribed interface, global properties like deadlock-freedom are still guaranteed. A drawback to session types, however, is that they are very verbose---users need to specificy the type of each channel, (the type and direction of data sent and received over each channel) and also implement each process over that channel, which does the sending and receiving. Furthermore, each channel has two types associated with each of the two endpoints, leading to a significant amount of repetition in the description of the protocol. Choreographies on the other hand, start with the global view and infer the communication information from that single perspective.

\subsection{Quantum network simulators and languages.}

Because near-term quantum network hardware is scarce, simulation is a central tool for protocol design and software development. 
Simulators model local gates and measurement as well as 
probabilistic entanglement generation, routing, and scheduling~\citep{kozlowski_designing_2020, coopmans_netsquid_2021}. Discrete-event simulators like NetSquid~\cite{coopmans_netsquid_2021} and SeQUeNCe~\cite{wu_sequence_2020} are used to model and understand the behavior of large quantum networks. Other simulator libraries like SimulaQron~\cite{dahlberg_simulaqron_2017} and QuNetSim~\cite{diadamo_qunetsim_2020} emphasize software development and protocol prototyping at a higher level. These tools make different tradeoffs: some prioritize detailed physical modeling, some scalability, and some a more convenient programming interface for application-level protocols.



NetQASM~\citep{dahlberg_netqasm_2022} is complementary to these simulators: it is a low-level instruction set architecture and SDK for hybrid classical/quantum programs that run on quantum network nodes. It is not itself a simulator, but an intermediate representation that is designed both to be interpreted by simulators and to be run on real quantum network architectures
In this work, we use NetQASM as a platform-agnostic intermediate representation that supports experimenting with projected Qoreo protocols.

\section{Conclusion and future work}
\label{sec:conclusion}
Qoreo opens up several avenues for exploring distributed quantum programming in the future.

\subsection{Expanded protocols and knowledge of choice}

In choreographic programming languages, knowledge of choice refers to the ability for different communication patterns to occur depending on dynamic decisions made by individual actors~\cite{montesi2023introduction}. To do this safely, the relevant actors must coordinate amongst themselves. Knowledge of choice is well-studied in classical choreographies, and its integration into Qoreo should be largely independent of the quantum-specific design choices made here. Adding knowledge of choice would enable Qoreo to represent an even richer class of choreographic protocols.

\subsection{Optimizations on choreographic distributed quantum programs}

Choreographies may also serve as an effective intermediate representation for optimization. Because a choreography presents a global view of the protocol, optimization passes can reason across abstraction layers: low-level quantum communication costs, actor-local computation structure, and global synchronization points. In particular, this perspective can expose opportunities to rearrange where different portions of a computation are performed, while still preserving protocol-level correctness constraints.

Expanding the communication model would further enrich this optimization space. Rather than restricting the language to binary entanglement generation alone, we could admit higher-level primitives such as teleportation and distributed unitary gates, such as the TeleData and TeleGate models of \citet{caleffi_distributed_2024}. These additional primitives could support optimizations that are difficult to express at a single abstraction level, including Bell-state usage minimization and circuit-cutting-style decompositions that trade communication for local computation.

\subsection{Quantum networks, noise, and more}

There are many different models of quantum distributed systems beyond the one presented in this paper. For example, we might consider heterogeneous networks where different nodes have different capabilities, networks with limited connectivity between nodes, or local systems that distinguish communication qubits from computation qubits. Extending Qoreo to handle these more realistic network models could require a heterogeneous collection of local quantum languages, each tailored to the capabilities of individual node types.

One important consideration outside the scope of this work involves noise. Current quantum computers are extremely noisy: every quantum operation carries a chance of error, and even an idle system decoheres over time. Reasoning about accumulated error rates and decoherence windows may be more tractable at the choreographic level than at the process level, since the global view makes communication depth and operation placement explicit. This perspective could support both static analysis of noise budgets and automated transformations that reduce communication depth or consolidate operations to minimize exposure to decoherence.

\subsection{Conclusion}

Qoreo enables distributed quantum programming through a global, choreographic view that can be soundly projected onto asynchronous processes. This paper formally verifies the correctness of Qoreo's meta-theory and demonstrates through examples that choreographies naturally and concisely capture distributed quantum protocols. We hope this work serves as a foundation for future tools and languages that bring the benefits of choreographic programming to distributed quantum computing.

\bibliographystyle{ACM-Reference-Format}
\bibliography{references}


\end{document}